\def\year{2019}\relax
\documentclass[letterpaper]{article} 
\usepackage{aaai19}
\usepackage{hyperref}
\usepackage{times}  
\usepackage{helvet}  
\usepackage{courier}  
\usepackage{graphicx}  
\usepackage{amsthm}
\newcommand{\citet}[1]{\citeauthor{#1}~\shortcite{#1}}
\newcommand{\citep}{\cite}
\newcommand{\citealp}[1]{\citeauthor{#1}~\citeyear{#1}}
\usepackage{amsfonts}
\usepackage{booktabs} 
\usepackage{graphics}
\usepackage{nicefrac}      
\usepackage{microtype}      
\usepackage{latexsym}
\usepackage{amsmath}
\usepackage{bm}

\usepackage{algpseudocode}
\usepackage{todonotes}

\newcommand{\beq}{\begin{eqnarray*}}
\newcommand{\eeq}{\end{eqnarray*}}
\newcommand{\beqn}{\begin{eqnarray}}
\newcommand{\eeqn}{\end{eqnarray}}

\newcommand{\hide}[1]{}

\newcommand{\hp}{\hat{\pi}}
\newcommand{\deviation}{\mathrm{Deviation}}

\newcommand{\btrain}{\bar{B}}

\newcommand{\bepf}{\begin{proof}}
\newcommand{\enpf}{\end{proof}}

\usepackage{enumitem,todonotes}
\usepackage{algorithm}
\usepackage{pifont}
\usepackage{verbatim}
\usepackage{tikz}
\usetikzlibrary{arrows}
\usetikzlibrary{calc}
\usetikzlibrary{shapes}
\usetikzlibrary{fit}
\usetikzlibrary{matrix} 
\usetikzlibrary{positioning}
\usetikzlibrary{decorations.pathreplacing}
\newtheorem{theorem}{Theorem}[section]

\newcommand{\figref}[1]{Figure \ref{fig:#1}}
\newcommand{\tabref}[1]{Table \ref{tab:#1}}        
\newcommand{\secref}[1]{Section \ref{sec:#1}}

\newcommand{\algoref}[1]{Alg.~\ref{alg:#1}}

\newcommand{\foldlogl}{\mathrm{FoldedLL}}
\def\P{\mathbb{P}}
\newcommand{\E}{\mathbb{E}}

\newcommand{\reals}{\mathbb{R}}

\newcommand{\half}{{\frac12}}

\newcommand{\one}{\mathbb{I}}
\newcommand{\norm}[1]{\|#1\|}
\newcommand{\dotprod}[1]{\langle #1 \rangle}

\DeclareMathOperator*{\argmin}{argmin}
\DeclareMathOperator*{\argmax}{argmax}

\newcommand{\trn}{^{\mathsf{T}}}
\newcommand{\svd}{\operatorname{SVD}}
\newcommand{\logl}{\mathrm{LL}}

\newcommand{\fold}{_\mathrm{fold}}
\newcommand{\spn}{_\mathrm{sp}}
\newcommand{\tno}{_\mathrm{tr}}
\newcommand{\ber}{_\mathrm{Ber}}

\newcommand{\paren}[1]{\left( #1 \right)}
\newcommand{\pt}{\pi_t}
\newcommand{\pnull}{{\bar\pi}}
\newcommand{\phat}{{\hat\pi}}
\newcommand{\mse}{\mathrm{MSE}}

\newcommand{\cD}{\mathcal{D}}

\newcommand{\dataandcode}{\url{https://github.com/eyalgut/TLR\_anomaly\_detection.git}}
\newcommand{\kaggleurl}{\url{https://www.kaggle.com/eyalgut/binary-traffic-matrices}}
\nocopyright
\title{Temporal Anomaly Detection: Calibrating the Surprise}

\author{Eyal Gutflaish,\textsuperscript{1}
Aryeh Kontorovich,\textsuperscript{1}
Sivan Sabato,\textsuperscript{1}  
Ofer Biller,\textsuperscript{2}
Oded Sofer\textsuperscript{2}\\
\textsuperscript{1} Ben-Gurion University of the Negev, Beer Sheva, Israel\\
\textsuperscript{2} IBM Security Division, Israel \\
eyalgutf@gmail.com,
\{karyeh, sabatos\}@cs.bgu.ac.il,
\{ofer.biller, odedso\}@il.ibm.com}

\begin{document}
 
\maketitle

\begin{abstract}
We propose a hybrid approach to temporal anomaly detection
in access data of users to databases --- or more generally, any kind of subject-object co-occurrence data. We consider a high-dimensional setting that also requires fast computation at test time. Our methodology identifies anomalies based on a single stationary model, instead of requiring a full temporal one, which would be prohibitive in this setting. We learn a low-rank stationary model from the training data, and then fit a regression model for predicting the expected likelihood score of normal access patterns in the future. The disparity between the predicted likelihood score
and the observed one is used to assess the ``surprise'' at test time.
This approach enables calibration of the anomaly score, so that time-varying normal behavior patterns are not considered anomalous. We provide a detailed description of the algorithm, including a convergence analysis, and report encouraging empirical results. One of the data sets that we tested, TDA, is new for the public domain. It consists of two months' worth of database access records from a live system. 
Our code is publicly available at \dataandcode. The TDA data set is available at \kaggleurl.
\end{abstract}

\newcommand{\tableswitch}{\begin{table}[h]

{\centering
		\resizebox{\columnwidth}{!}{%
\begin{tabular}{ccccccc}
      	Data Set  & TLR & Mean &SRMF& PCA&RobustPCA& MET  \\
\toprule
	
TDA & $\mathbf{39}\%$& $7\%$ &-&-&- &-\\
Amazon   & $\mathbf{84}\%$& $5\%$&-&-&-&- \\
Netflix & $\mathbf{49}\%$ & 8\% &-&-&- &-\\
Movielens  & $\mathbf{16}\%$& $8\%$&-&-&- &-\\
	TDA (small)  &$\mathbf{22}\%$& $3\%$ & $6\%$&$5\%$ &$6\%$ &$5\%$ \\
	Amazon (small) &$\mathbf{99}\%$& $7\%$ & $82\%$ &$5\%$ &$7\%$ & $7\%$ \\
\bottomrule\\
\end{tabular}}
\caption{Percent of tested intervals which were assigned an anomaly score above 0.95 of the full and sampled date sets.}
\label{tab:switch}
}
\end{table}
}

	\newcommand{\figurenoise}{\begin{figure}[h] 
			{\centering
				\includegraphics[width=0.23\textwidth]{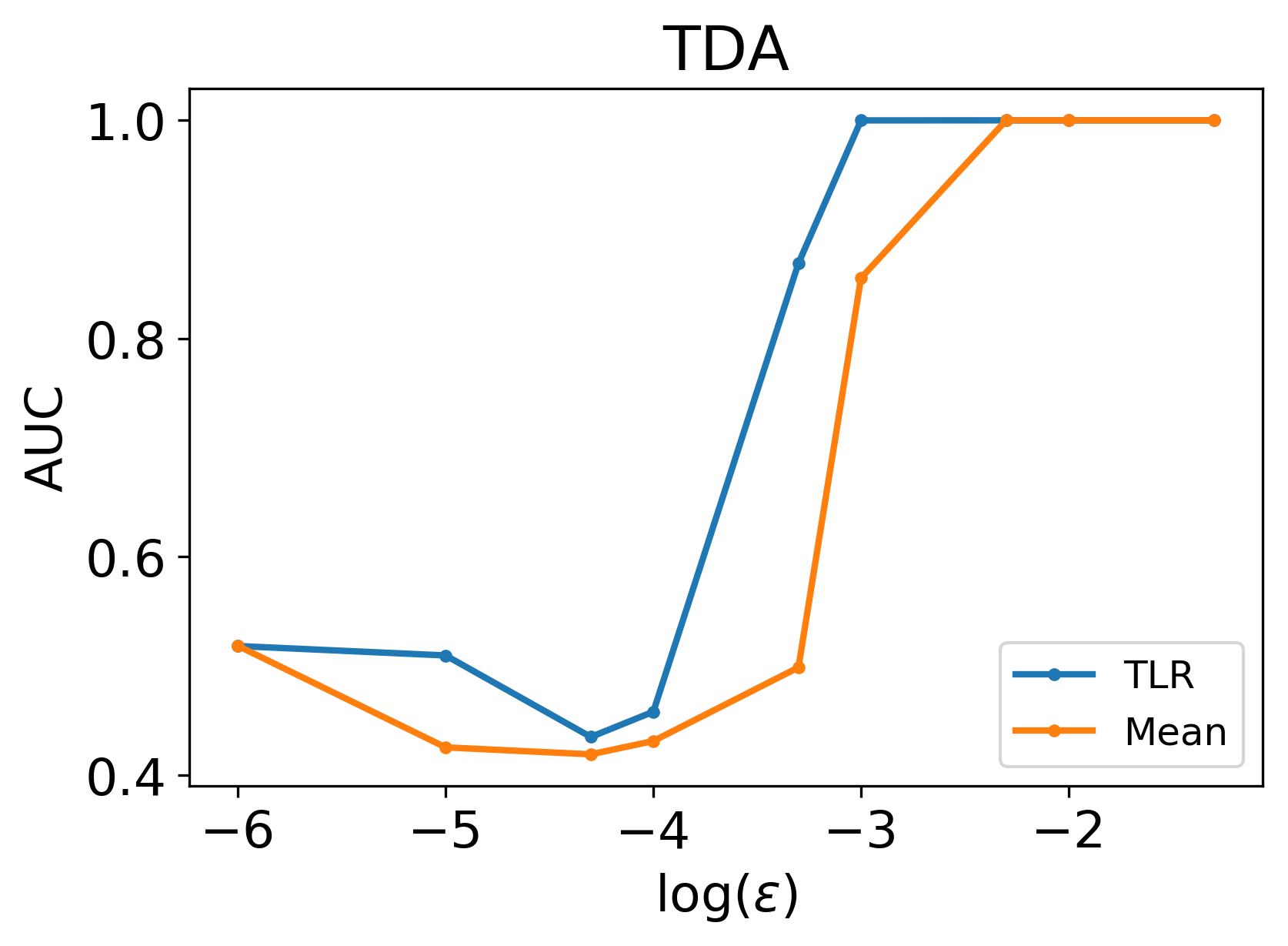}
				\includegraphics[width=0.23\textwidth]{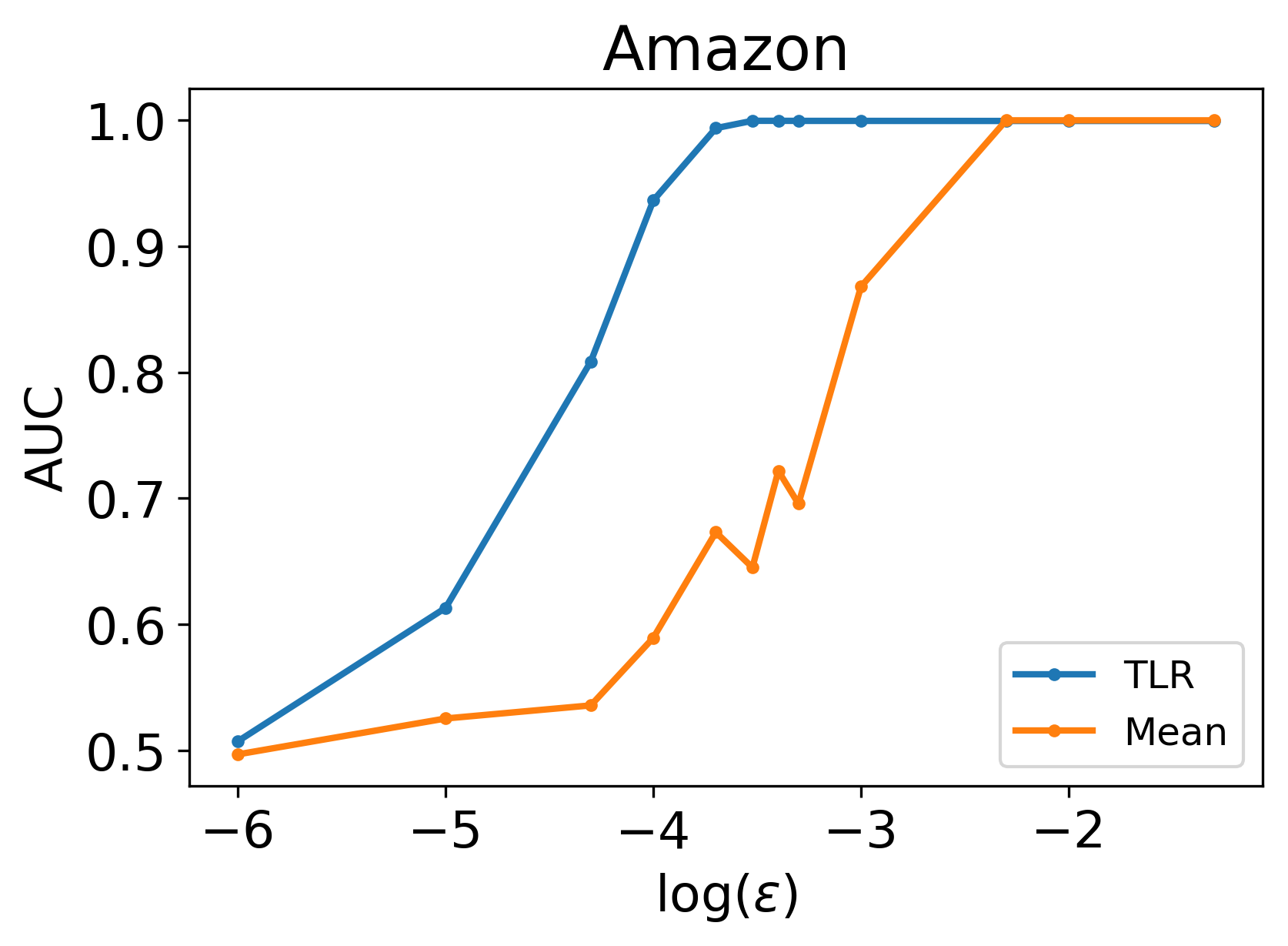}\\
				\includegraphics[width=0.23\textwidth]{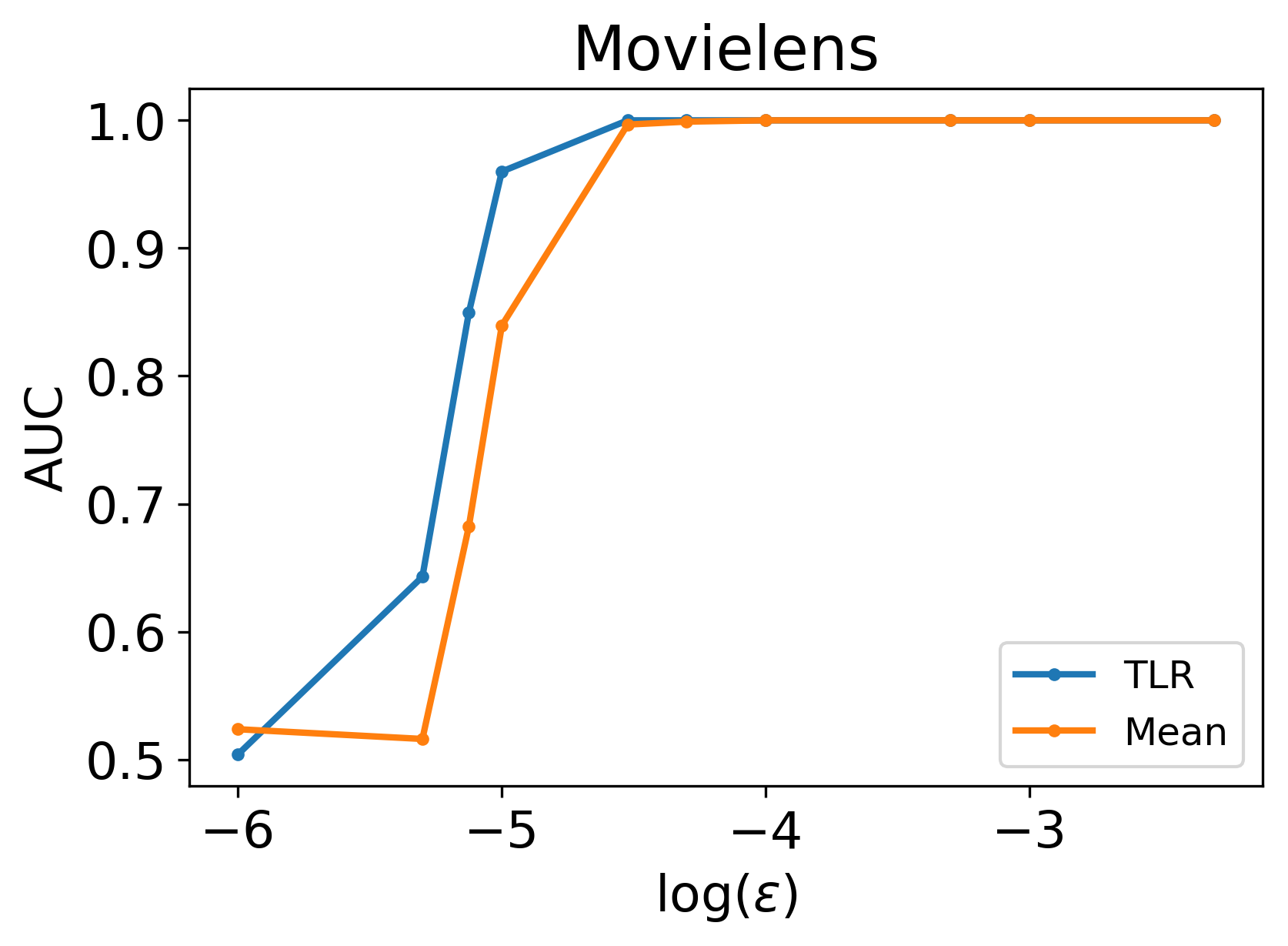}
				\includegraphics[width=0.23\textwidth]{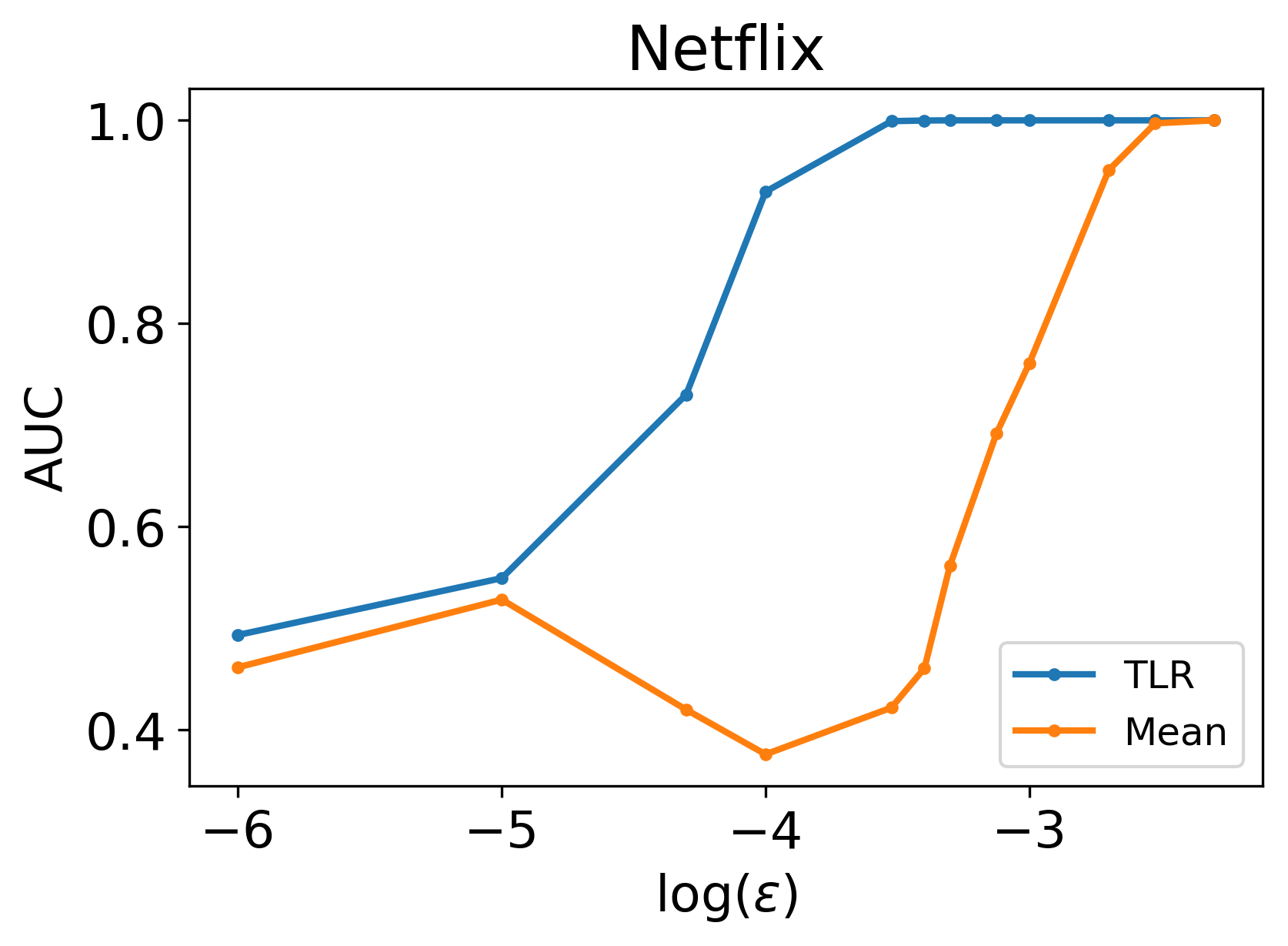}\\
				\includegraphics[width=0.23\textwidth]{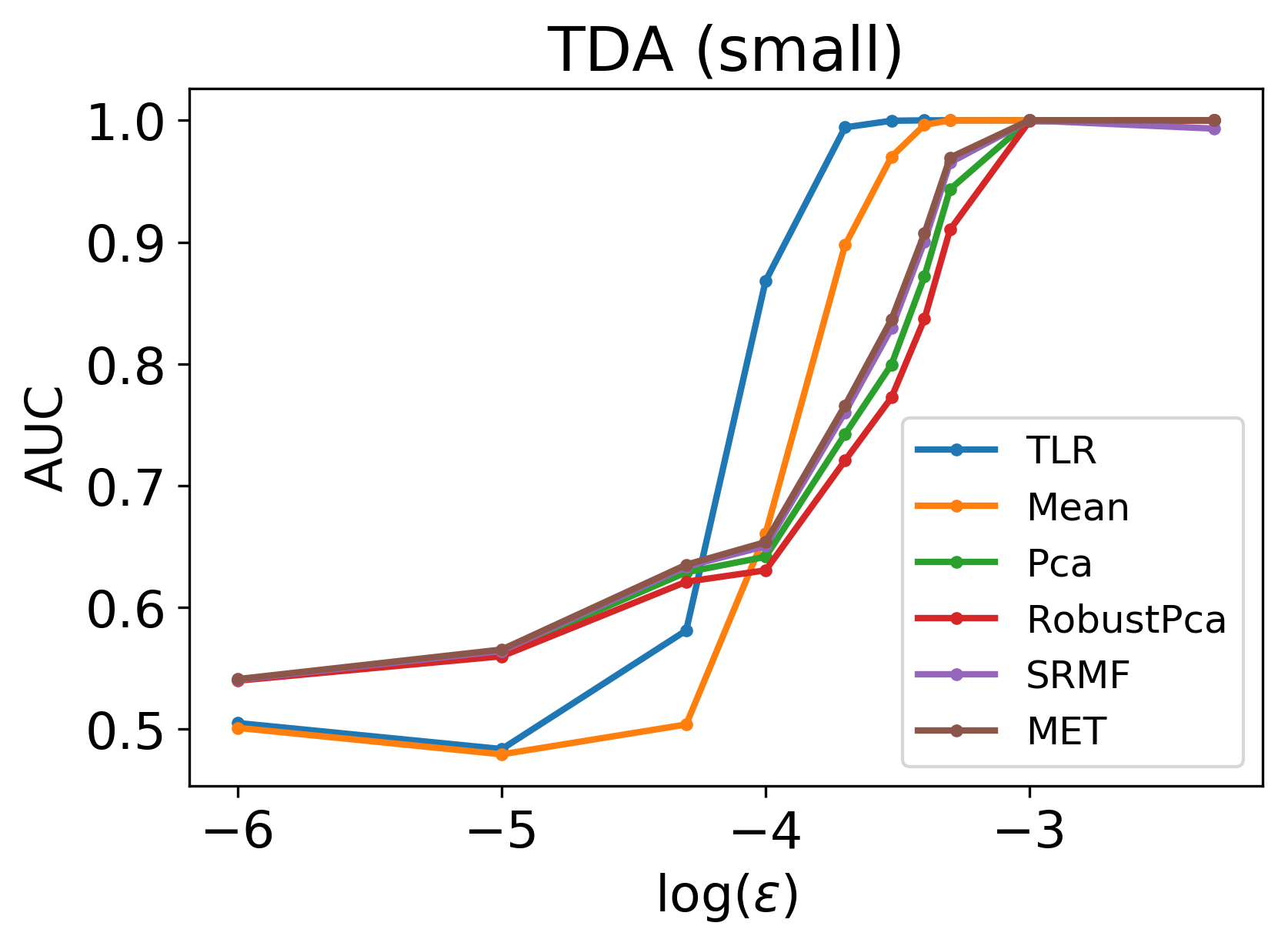}
				\includegraphics[width=0.23\textwidth]{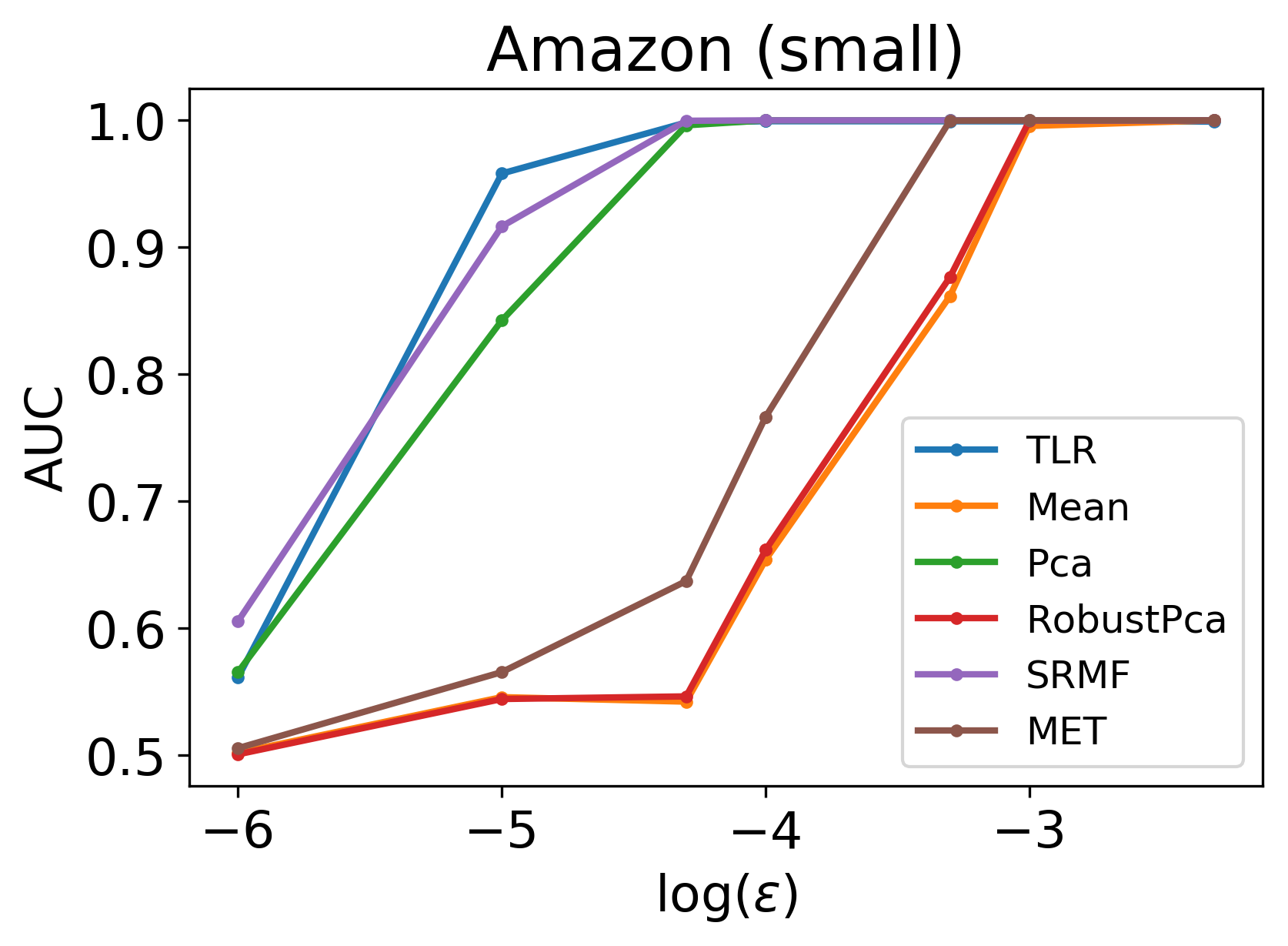}
			}
			\caption{The AUC as a function of the log of the noise in the random access experiment. 
				}
			\label{fig:noise}
		\end{figure}}
 
\newcommand{\figurehistswitches}{
	\begin{figure}[h]
 
			{\centering
				\includegraphics[width=0.23\textwidth]{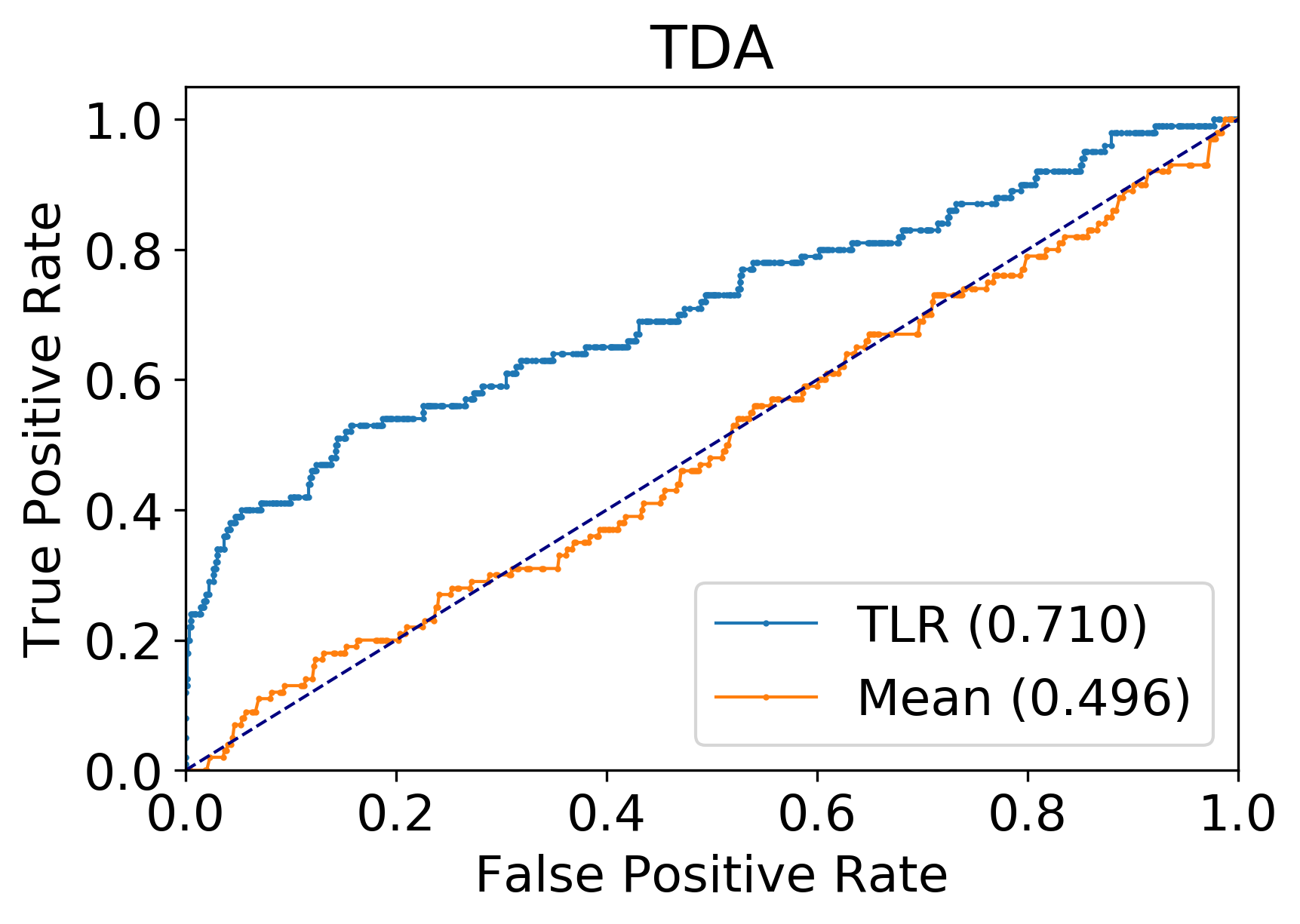}
				\includegraphics[width=0.23\textwidth]{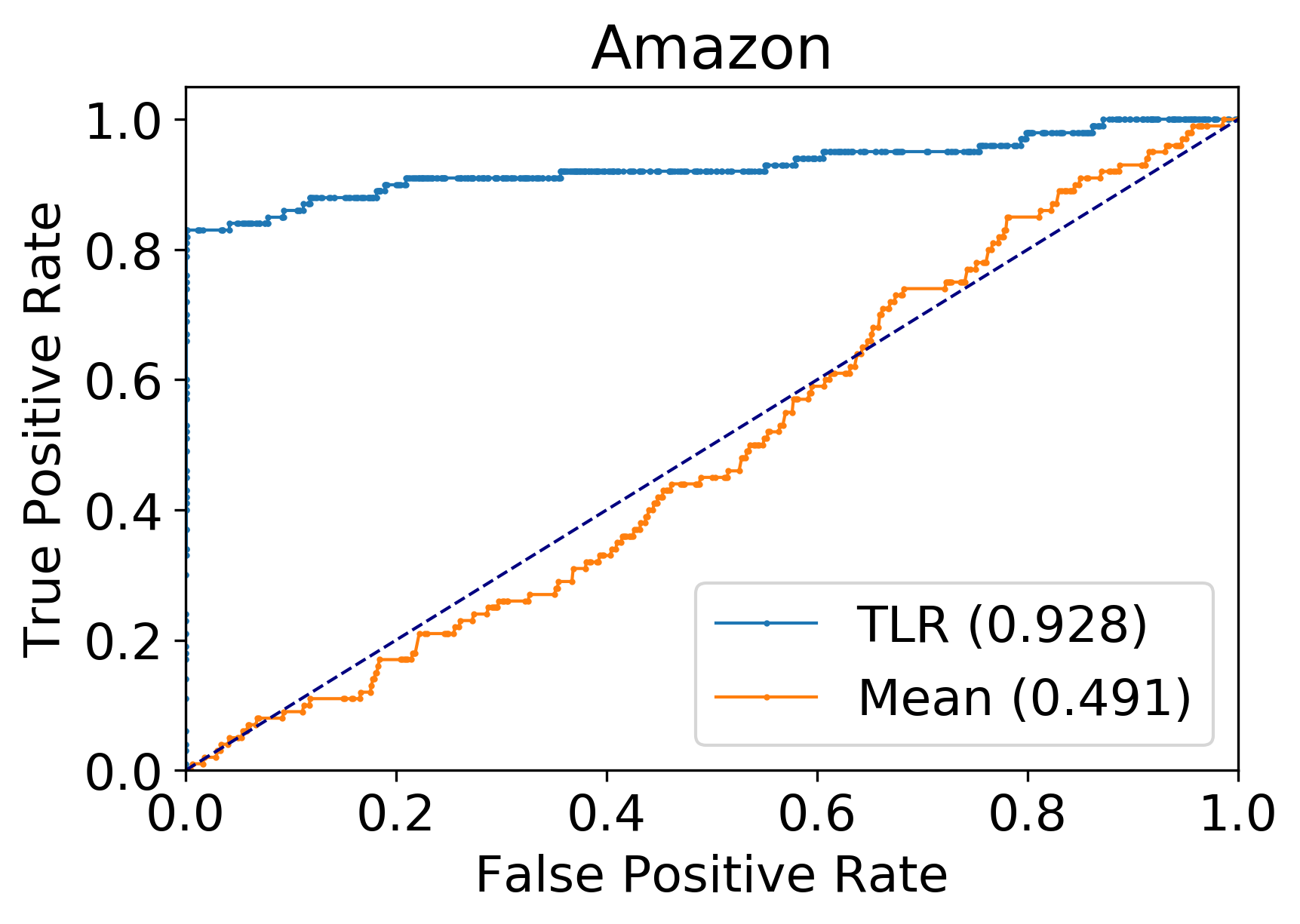}
				\includegraphics[width=0.23\textwidth]{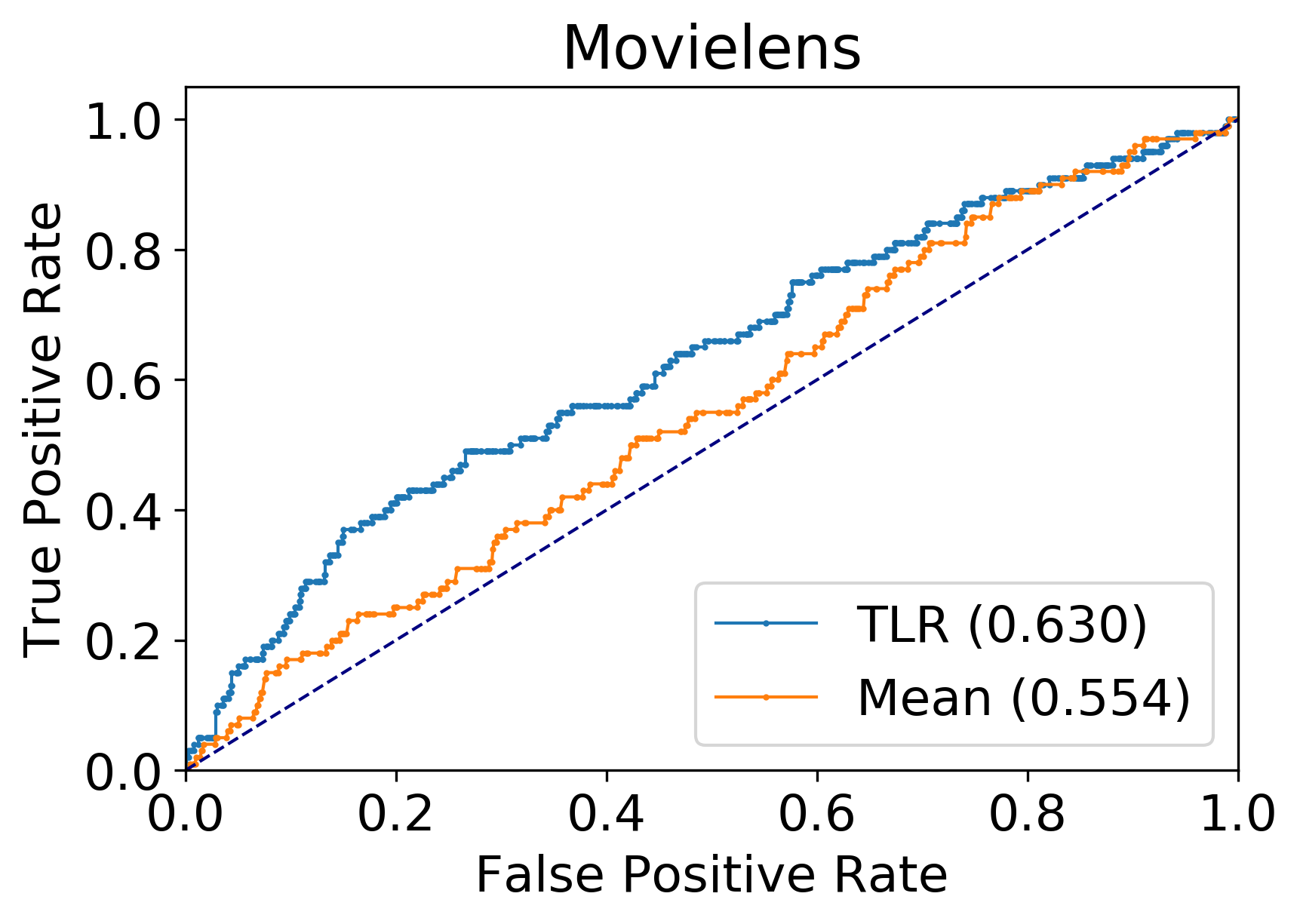}
				\includegraphics[width=0.23\textwidth]{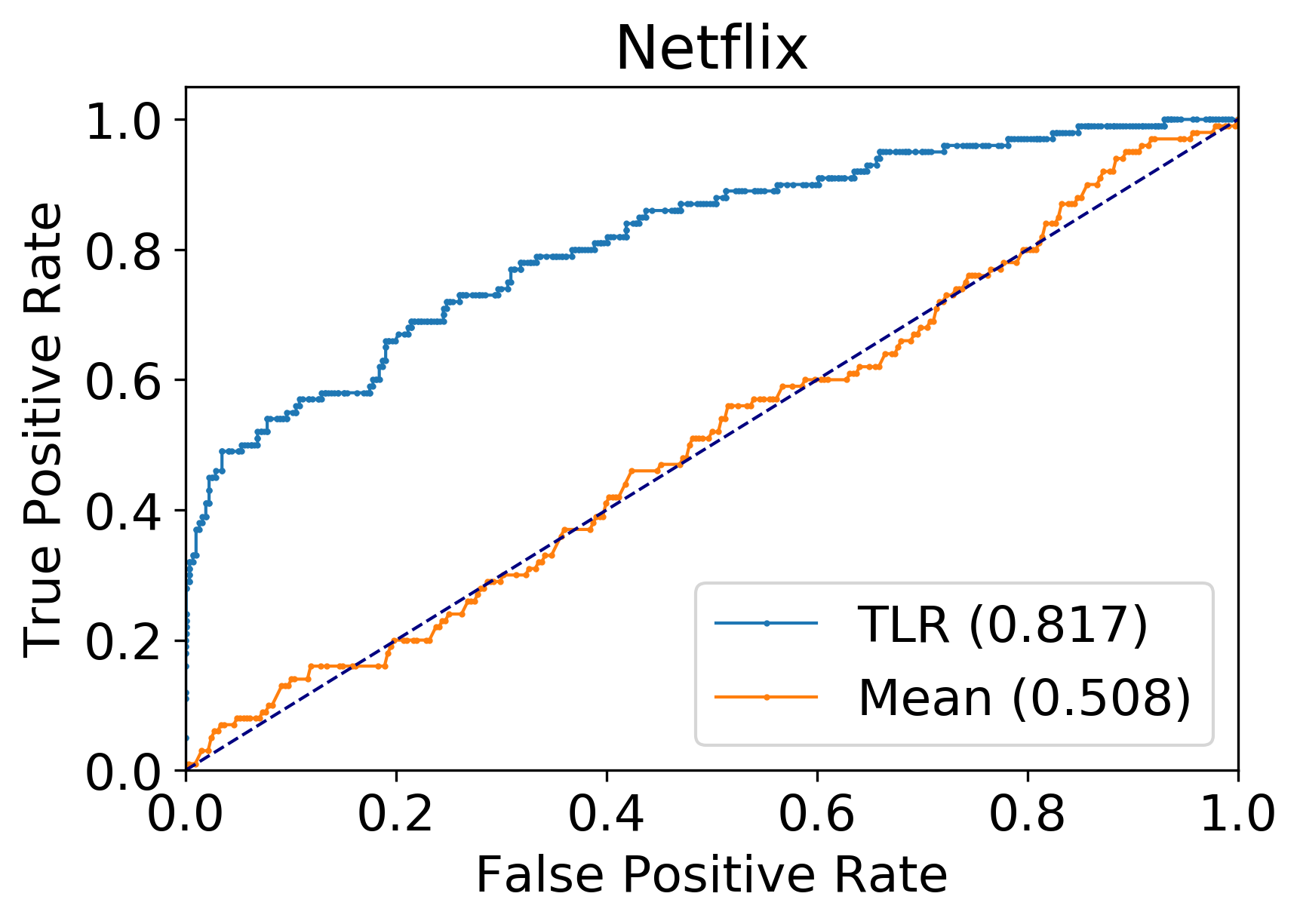}\\
				\includegraphics[width=0.23\textwidth]{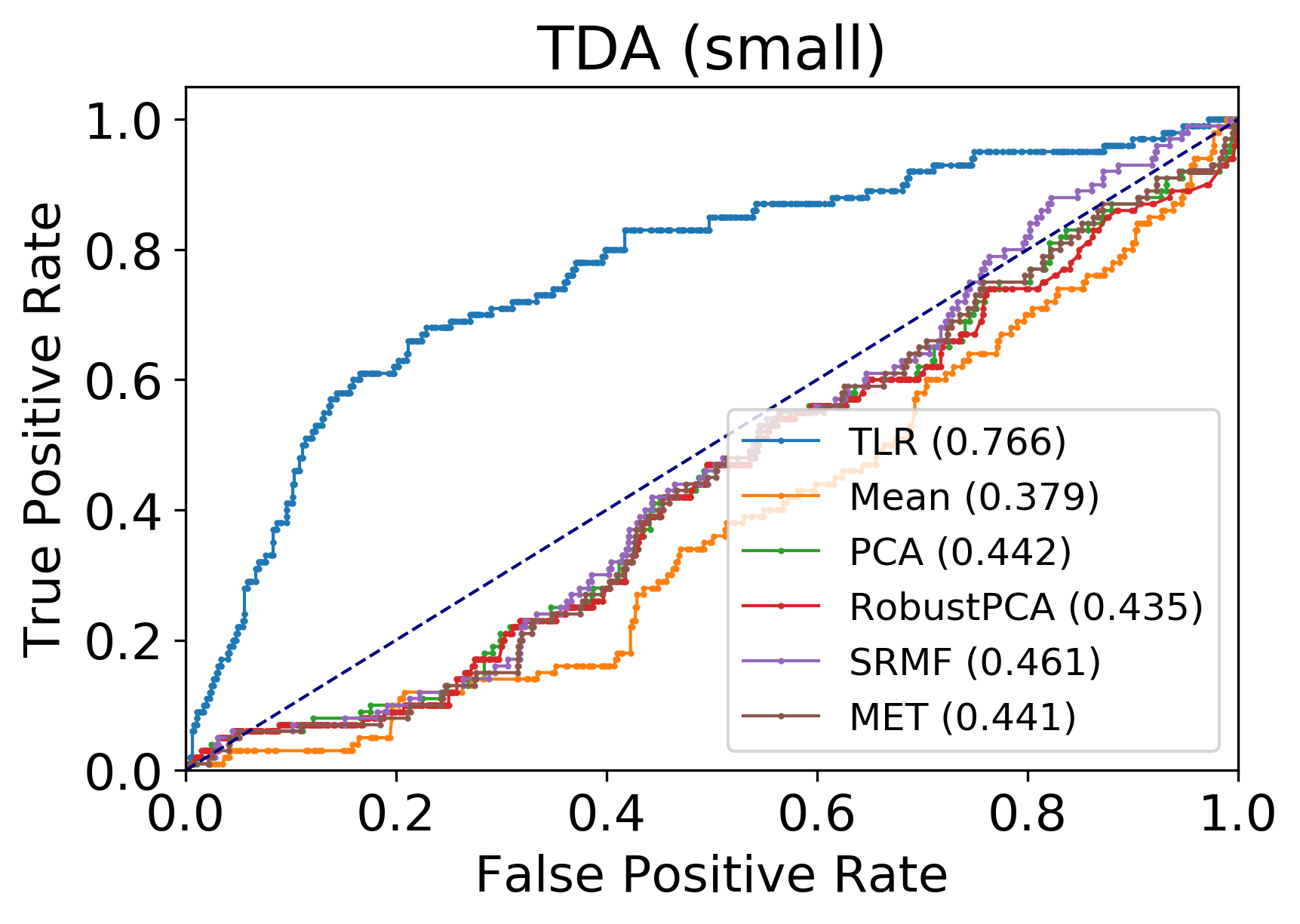}
				\includegraphics[width=0.23\textwidth]{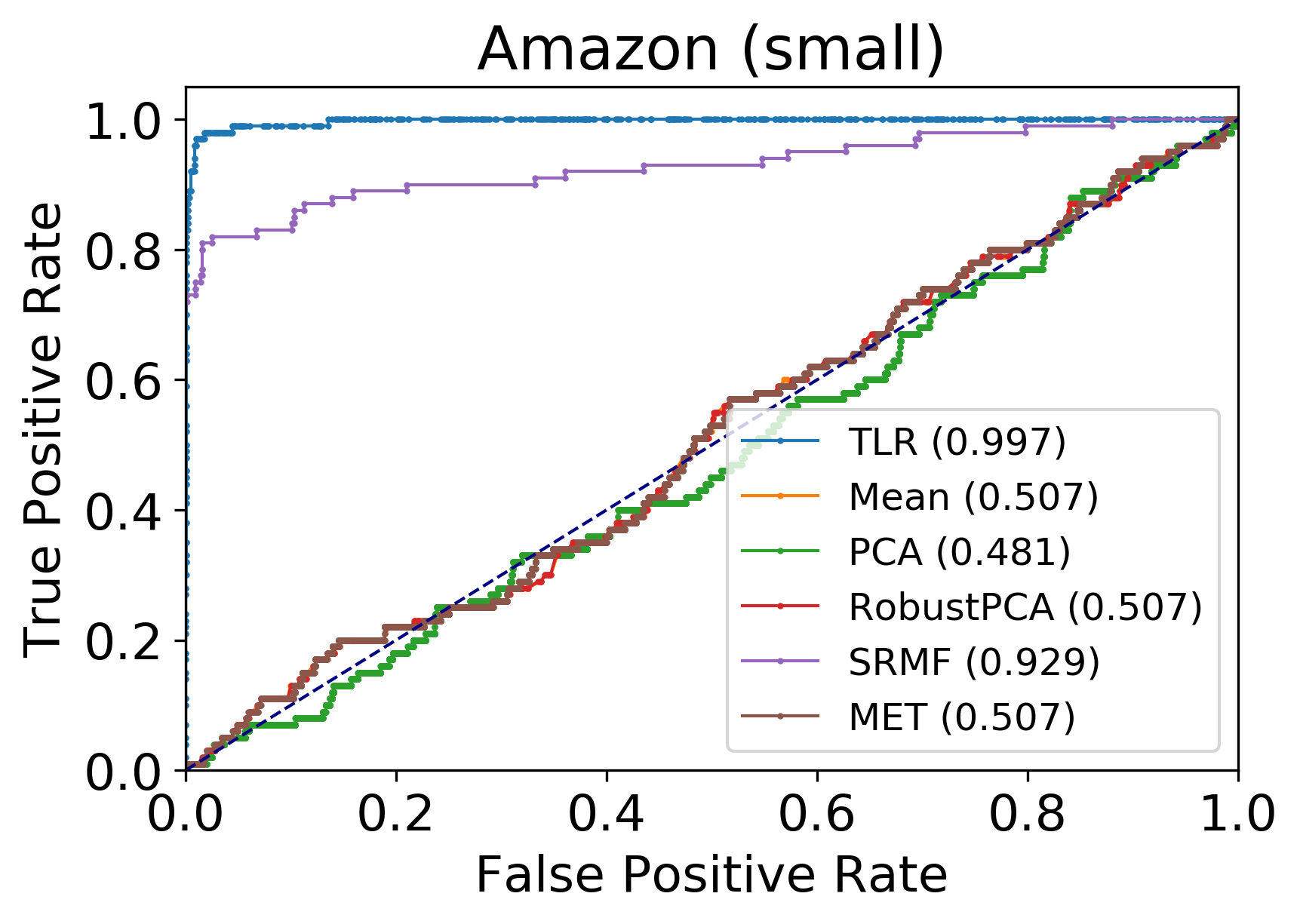}
			}
				\caption{ROC curves for the anomalous time experiment. The AUC is given in parentheses.}
		 	\label{fig:rocswitches}
		\end{figure}	
}
	
\newcommand{\tableruntime}{
\begin{table}[h]
		\resizebox{\columnwidth}{!}{
		\begin{tabular}{ccccccc}
			Data Set  & TLR (train) & TLR (test) &SRMF& PCA & RobustPCA & MET \\
			\toprule
				TDA  &  $\mathbf{946}$ & $\mathbf{0.0029}$ &- & - & - & - \\
					Amazon  &  $\mathbf{16621}$ & $\mathbf{0.0141}$ &- & - & -& -  \\
						MovieLens  &  $\mathbf{4652}$ & $\mathbf{0.0024}$ &- & - &-& -\\
							Netflix &$\mathbf{1964}$&$\mathbf{0.0071}$& - & - &  -& -  \\
			TDA (small) &  $\mathbf{45}$ & $\mathbf{0.0042}$ &229 & 140 & 338 &50  \\
			Amazon (small) & $\mathbf{209}$ &$\mathbf{0.0027}$ &279& 571 & 451 & 354  \\
			\bottomrule\\
		\end{tabular}}
		\caption{Run-time (seconds) on an 2.8GHz Xeon CPU with 40 cores and 256 GB RAM. }
		\label{tab:timeAnalysis}
 \end{table}
}

\section{Introduction}\label{sec:intro}
Consider a security analyst examining
user
access logs of a large database system. A blatant
security breach
might involve a user with
insufficient clearance attempting to access a restricted database table.
However, there could be more subtle indicators of
suspicious
activity, such as users accessing database tables that are atypical
of
their past behavioral pattern,
or at unusual times.
Moreover, in a distributed attack, perhaps no single user has done anything
particularly
out of the ordinary,
but the general
pattern of access
to different database tables is
atypical in terms of frequency, time of day, identity of the users involved, and so forth. 

The difficulty stemming from the nebulous definition of
potentially suspicious activity
is compounded by the fact that
severe
anomalies, by their very nature,
are extremely rare occurrences, and when the goal is to learn to identify them,
we generally do not expect the available training data
to contain any positive examples.
Furthermore, due to the large number of users and database tables in a typical system,
a naive
solution, which classifies all
previously unseen access events as anomalies,
will tend to trigger
many false alarms.
The latter issue is exacerbated further by the problem of
{\em cold start}
 \citep{A_Survey_of_Collaborative_Filtering_Techniques}
--- that is, the activity of previously unseen users (say, new employees).
This activity should not automatically be classified as an anomaly, or else too many false alarms will be issued.
Additionally, unseen accesses invoked by known users and tables could also cause a similar problem.
Thus, a key difficulty in anomaly detection of temporal events from a complex system
is to calibrate the surprise level associated with incoming events
--- and this is the central challenge that we address.

\paragraph{Problem description.} We address the problem of
unsupervised anomaly detection in a high-dimensional temporal sequence of user-object access
events. The events might be company employees accessing database tables, users
interacting with a website, customers
executing
transactions, and so on. We assume that the users and the objects are atomic (that is, known by a name or identification number only).
Time is
discretized into fixed units (such as hours), and for each time unit, the
access activity is recorded in a
binary
user-object incidence matrix.

Our goal is to develop an anomaly detection method which allows identifying distributed attacks. Such attacks cannot be characterized by a
single suspicious access event --- the latter can be handled by more direct means, such as a permissions
system; instead, they are characterized by system-wide suspicious access patterns. Thus, our goal is to detect anomalous time intervals, i.e., segments of activity that contain atypical behavior. Our proposed method can also be adapted
to
diagnosing
a specific user's behavior,
or a specific object's access pattern,
as anomalous within a time interval. We leave
the details of this extension
to future work.

Our setting is similar to the problem of anomaly detection in sequences of \emph{traffic matrices} in communication engineering \citep{Roughan:2012:SCS:2369156.2369159},
where typically the traffic matrix of each time unit contains the number of packets (or bytes)
transferred from source to destination IP.
In our problem, we consider more general user-object incidence matrices
(but, being binary, these only register the presence or absence of an access). In addition, we assume high-dimensional data, and require a fast and efficient computation at test time. 

From a bird's eye view, our approach consists of two orthogonal components:
(1) fitting a generative model of the data based on a training set and
(2) assigning an anomaly score to new time segments based on the model that we learned.
For (1), given the high number of users and objects, as well as their atomicity,
we propose a model based on a low-rank assumption.
This allows decomposing users and objects into latent factors
\citep{Generalization_Error_Bounds_Collaborative_Prediction}, and discovering abnormal behavior patterns based on the latent factors of a user or object.
We extend to our setting previous theoretical results on the sample complexity of learning such a model. 
Several previous works on anomaly detection for traffic matrices, which we discuss in \secref{rel}, use low-rank models; however, in these works the low-rank model is over the space of (user-object pairs) $\times$ (time intervals), while our low-rank model is over the significantly lower-dimensional space of users $\times$ objects. As a result, our approach can handle higher-dimensional data sets, while requiring less working memory. In addition, our approach performs fast real-time anomaly detection after the model-training phase, while previous approaches do not distinguish between the phases---resulting in a heavy computational burden during the anomaly detection phase.

Our main methodological innovation is in step (2).
Here, the likelihood assigned to an observed behavior pattern,
based on the fitted model,
is used to calculate an anomaly score.
We learn the expected likelihood score from the training data,  as an independent
regression problem,
and then use the disparity between the
{\em predicted}
and
{\em observed}
likelihood scores 
as a measure of surprise.
This compensates for
characteristic
discrepancies
between the learned model
and the
observed
behavior at certain times.
For instance, a different activity pattern is expected during
different hours of the day.
However, learning separate models for each possible activity type is prohibitive,
for both statistical and computational reasons.
Our approach enables us to
learn a single model,
yet adapt the anomaly score to
time-varying
normal behavior patterns. 

This work was motivated by studying 
a real database monitoring system which performs anomaly detection.
We provide a new data set, which we call TDA (Temporal Database Accesses).
It
was generated by recording
user
accesses in a live database system
(of the kind often
monitored for anomalous behavior)
over a two-month period. The data set records accesses of thousands of users into thousands of database tables, and the rate of events was approximately $20,\!000$
per hour.
The data is provided in the form of binary access matrices,
indicating whether each user accessed each table during each 
one-hour-long interval. 
Our code is publicly available at \dataandcode. The TDA data set is available at \kaggleurl.

\section{Related Work}\label{sec:rel}

Intrusion detection methods can be roughly clustered into two basic categories:
rules-based and learning-based methods
\citep{journals/sigmod/SantosBV14}.
Some approaches \citep{oai:CiteSeerPSU:231451,journals/vldb/KamraTB08,journals/jcp/SrivastavaSM06,conf/dbsec/SpalkaL05,conf/raid/MathewPNU10,LeeLowWon02} require full sequential event information. In contrast, in this work we focus on the case where accesses in each interval are aggregated, and sequence data is not available.

The problem of change-point detection \citep{takeuchi2006unifying,tartakovsky2006novel,hohle2010online,khaleghi2014asymptotically}, 
while not strictly a subset of anomaly or intrusion detection,
is of some relevance
to the temporal setting. Another natural approach is to model the temporal process via a Markov Model or a Hidden Markov Model,
as was done in
\citet{DBLP:conf/icml/GornitzBK15,Soule:2005:TMT:1111572.1111580}. These approaches, however, are infeasible in our setting: due to the large number of users and objects that we are dealing with (typically in the thousands), billions of parameters would need to be estimated, which is impossible to do from a reasonable amount of data. 

Variants of association
rules are used in \citet{chan2003machine,das2007detecting,das2008anomaly};
these are 
less suitable for handling new instances and large spaces. 
Some supervised and semi-supervised approaches have also been suggested \citep{Gunnemann:2014:DAD:2623330.2623721,5694074,4215661}. These are applicable when there is supervision data on anomalies.

The problem of anomaly detection in sequences of traffic matrices seems the most similar to our setting. The traffic matrix of each time unit contains the amount of data transferred from
 source IP address to destination IP address. 
\citet{Roughan:2012:SCS:2369156.2369159} propose the SRMF (Sparsity Regularized Matrix Factorization) algorithm.
A
matrix is constructed by vectorizing the traffic matrices
into
the
columns of a new matrix,
whose width is the number of time intervals in the data set,
and whose height is square the number of IPs in the data set.
Each row in this matrix corresponds to a single IP-IP pair, and each column corresponds to a single time interval.
Temporal smoothness
is assumed.
A smooth
and low-rank approximation is obtained for this matrix via regularized optimization,
and used as a
baseline to detect
anomalous time intervals.
In our setting, SRMF can be used by replacing the vectorized (IP address)$\times$(IP address) matrices with
vectorized user $\times$ object access matrices.
\citet{Lakhina:2004:DNT:1015467.1015492} construct a similar matrix, but obtain a low-rank model using PCA. 
\citet{Candes:2011:RPC:1970392.1970395} propose a Robust PCA algorithm for finding anomalies in a sequence of images, again using a similar matrix with vectorized matrices (images) as columns. These methods all look for low-rank structures  in the space of (user-object pairs) $\times$ (time intervals).
\citet{TTB_MET} propose MET, a tensor decomposition technique, which searches for a low rank structure in the space of users $\times$ objects $\times$ time-intervals. We compare our algorithm to these approaches in our experiments in \secref{experiments}. \citet{zhou2017anomaly} and \citet{azzouni2017long} propose deep-learning approaches for anomaly detection on matrices. These methods require a fully-connected input layer, which in our case would include millions of features, and are not applicable to the high-dimensional data sets that we study here. 

\section{Our approach}\label{sec:approach}	
To detect anomalous access patterns, we define a probabilistic model for normal access patterns.
We learn a baseline low-rank stationary model for a user-object incidence matrix, and then model the deviation of the
temporal model from the stationary one. This 
enables
learning and detection using a feasible number of parameters.

Denote the number of different users by $n$ and the number of different objects by $m$. For simplicity of notation we fix $m$ and $n$; however, in practice they need not be known to the algorithm in advance. 
We assume that the data is provided as a sequence of consecutive time intervals, where for each time interval $t$  an access matrix $B_t \in \{0,1\}^{n \times m}$ is provided, where $B_t(i,j) = \one[\text{user $i$ accessed object $j$}]$. The length of a time interval is
an external application-specific parameter.
The goal of the algorithm is to
assign an anomaly score to each new access matrix $B_t$ which is observed after the training phase. 
The distribution of $B_t$
could be
modeled using a matrix $\pt\in [0,1]^{n \times m}$, where $\pt(i,j)$ is the probability 
that user $i$ accesses object $j$ during time interval $t$, and different entries in $B_t$ are assumed statistically independent. 
Thus, at time interval $t$, any possible observation matrix $G \in  \{0,1\}^{n \times m}$
would be
assigned a probability of
\begin{align}
  \label{gen-model-naive}
  \P_{\pt}[B_t = G]
  :=
\!\!\!\!  
  \prod_{i \in [n],j\in [m]}
\!\!\!\!
  \pt(i,j)^{G(i,j)}(1-\pt(i,j))^{1-G(i,j)}.
\end{align}
This model
allocates
a separate set of parameters for each time interval,
and
is incapable of
extrapolating beyond
past observations.
Hence, we instead posit
a single baseline matrix $\pnull \in [0,1]^{n \times m}$, which approximates a stationary (time-independent) distribution. This baseline matrix can be thought of as a
rough
approximation of $\pi_t$ for all time intervals $t$.
It induces a distribution on observation matrices
in a manner analogous to \eqref{gen-model-naive}:
$\P[B_t = G] :=
\P_{\pnull}[B_t = G]$.
This model is similar to the one proposed in \citet{davenport20141} for a non-temporal
variant
of matrix completion from probabilistic binary
observations.
We take the standard approach of assuming that $\pnull$ is low-rank, motivated by the intuition that the relevance
of a user to an object
can be explained by a small number of latent factors
(see, e.g., \citealp{A_Survey_of_Collaborative_Filtering_Techniques,One-Class_Matrix_Completion_with_Low-Density_Factorizations,Mining_of_massive_datasets}).

Let $\hp$ be an estimator for $\bar{\pi}$, which is used to approximate $\bar{\pi}$. In \secref{lowrank} we give our procedure for obtaining $\hp$.
Having obtained an approximation $\hp$ to 
$\pnull$ based on the training set, we can calculate the log-likelihood of an
observation matrix $G$ at time-interval $t$,
as induced by the parameters
$\hat{\pi}$:
\begin{align*}
&\logl(G,\hp) :=\log \P_{\hp}[G]
=
\\&
\sum_{i,j}\left(
  G(i,j)\log\hp(i,j)+
  (1-G(i,j))
  \log(1-\hp(i,j))
  \right).
\end{align*}

At this point, one might consider assigning time-interval $t$ an anomaly score based on the value $\logl(B_t,\hp)$, where $B_t$ is the actual matrix observed at time $t$:
a lower log-likelihood value
would indicate a higher anomaly level.
The problem with this proposal is that
it is likely that some time intervals
will
systematically exhibit behavior that deviates significantly from that of $\pnull$,
and these systematic deviations should not be classified as anomalies. In fact, it would be completely normal for these deviations to occur, and less normal if they do not occur.
For instance, it is expected that access patterns should be different between night and day,
weekdays and weekends, holidays and workdays, and so on, as well as be affected by application-specific circumstances.
For instance,
if the application monitors a software company's database accesses, scheduled days of major
version
updates
would likely
have patterns different from other days.
Thus, we need some way of accounting for systematic, non-anomalous, differences between time intervals.

We address this issue by proposing a compromise between the
overly constraining stationary model defined by $\pnull$ and the overly rich model in
\eqref{gen-model-naive}.
We model
the similarity between $\pt$ and $\pnull$
in terms of the
properties of the time interval $t$. This similarity can
be formalized using the cross-entropy between $\pnull$ and $\pt$. 
Recall that the cross-entropy between two discrete
distributions
$p,q$ is $H(p,q) := -\sum_{i} p_i \log(q_i). $
For two distributions defined as above by matrices $\pi_1,\pi_2 \in [0,1]^{m \times n}$, we have 
\[
H(\pi_1,\pi_2) = \sum_{i,j}H\ber(\pi_1(i,j),\pi_2(i,j)),
\]
where $H\ber(a,b)$ for $a,b \in [0,1]$ is the cross-entropy between the distributions Bernoulli$(a)$ and Bernoulli$(b)$.
The expected value of the measured log-likelihood score $\logl(B_t, \hp)$ satisfies
\[
E_{B_t \sim \pt}[\logl(B_{t},\hp)] = -H(\pt,\hp).
\]
Therefore, if the actual log-likelihood score $\logl(B_{t},\hp)$ is far from $-H(\pt,\hp)$, this can be considered an anomalous behavior. 
\newcommand{\hkl}{\widehat{KL}}
The value of $H(\pt,\hp)$ cannot be computed, since $\pt$ is unknown. 
Instead, we train a predictor that estimates it. We represent each time-interval
$t$ by a vector $v_t \in \reals^d$ of $d$ natural time-dependent real-valued features, such as time-of-day, day-of-weak, auto-regressive features (such as the log-likelihood in a previous interval) and possibly application-specific features, such as a binary indicator for times of version updates in a software company.
We fit a linear regression model $v_t \mapsto \dotprod{\hat{w}, v_t}$ parametrized by $\hat{w} \in \reals^d$, where $\logl(B_{t},\hp) \approx \dotprod{\hat{w}, v_t}$ on the training set.\footnote{Central to this approach is the assumption that anomalous intervals are very rare, and so
  the model is trained
  almost exclusively
  on
  non-anomalous behavior.}
We then define the anomaly score of a new observed time interval $t$ by
\begin{equation}\label{dev}
\deviation(\hp, B_t, t) = |\logl(B_{t},\hp)- \dotprod{\hat{w}, v_t}|.
\end{equation}
This approach
enables
identifying anomalous behavior,
while avoiding many of the false alarms resulting from normal differences between time intervals. Our definition of deviation identifies cases of a high likelihood also as anomalous, since they might indicate that the interval is less noisy than expected, which might also indicate a possible issue.

Lastly, we address the issue of \emph{cold-start} \citep{A_Survey_of_Collaborative_Filtering_Techniques}, in which new users and objects can appear for the first time in the test set, without ever having appeared in the training set. 
For instance, in the database-access setting,
new employees and new database tables can be added over time. If the application monitors an open environment such as a public web site, then the users and objects modeled in $\pnull$ could be a small minority of the set of users and objects observed during the deployment of the system. 
We address this issue by applying a process commonly known as  {\em folding} (e.g., \citealp{deerwester1990indexing,MR1722790}) to incorporate the new users or objects into the model on the fly.
In the next section we give a detailed account of the full anomaly detection algorithm. 
	
\section{The Algorithm}\label{sec:algorithm}
We describe the two phases of the algorithm: training and testing. In the training phase the model is learned. In the testing phase new intervals come in and are assigned an anomaly score based on the learned model.

In the
training
phase, the algorithm receives a training set $S = (B_1,\ldots,B_T)$, of consecutive access matrices. 
We split $S$ into two parts, $S_1 = (B_1,\ldots,B_{T_1}),S_2 = (B_{T_1+1},\ldots,B_T)$. $S_1$ is used to find an estimator $\hp$ for the probabilistic stationary model $\pnull$, while $S_2$ is used to fit the log-likelihood regressor $\hat{w}$. The full training algorithm is described in \secref{full}. 

In the anomaly-detection (testing) phase, an access matrix is provided as input for each time interval, and the algorithm outputs an anomaly score for each such matrix using Eq.~\eqref{dev}.
The monitoring system
now has a ranking of all the intervals
by anomaly score, and it can display the full ranking or the top few,
as specified
by the desired user interface.
For instance, if the security analyst can study 10 events a day then the top 10 suspected anomalies will be presented. 
Thus the threshold of anomalies to display depends on the capacity of the security analyst and the definitions of the monitoring system.

For simplicity of presentation, 
We first describe the two phases of the algorithm assuming that no new users or objects appear after the model $\hp$ is estimated in the training phase. We then explain, in  \secref{coldstart}, how the algorithm is seamlessly adapted to handle new users or objects.

\paragraph{Computational complexity} The most computationally expensive step in the algorithm, which we detail below, is an SVD procedure.
A naive implementation of SVD is can be
cubic in the matrix dimensions. However, 
since in our application the matrices are usually sparse,
a sparse SVD algorithm can be used to speed up computation (e.g., \citealp{larsen2000computing}).
All other procedures that our algorithm employs are linear in $m$ and/or $n$. 

Below we use several matrix norms: for a matrix $A \in \reals^{m \times n}$, denote the nuclear (trace) norm by
$\norm{A}\tno = \sum_{i}\sigma_{i}$, where $\sigma_i$ are the singular values of $A$. The Frobenius norm is $\norm{A}_F = (\sum_{i=1}^m \sum_{j=1}^n A_{i,j}^{2})^{\half}$, and the spectral norm is $\norm{A}\spn = \max \sigma_i$.

\subsection{Estimating the matrix model}\label{sec:lowrank}
Our probabilistic estimation problem in the first part of the training process is to estimate a low-rank probability matrix $\hp$ based on the sequence of matrices $S_1 = (B_1,\ldots,B_{T_1})$. 
In our simplified probabilistic model,
the
$B_t$'s are assumed to be drawn i.i.d.~according to some low-rank matrix $\pnull$. 

A standard approach for finding a low-rank estimate \citep{fazel2002matrix}
is to minimize the mean-squared error of the matrix difference, and regularize using the trace norm,
which is a convex relaxation of the low-rank constraint. Previous works assume that only a single access matrix drawn from $\pnull$ is available \citep{davenport20141,PU_Learning_for_Matrix_Completion}, while in our setting several matrices are provided at training time. To estimate $\pnull$, we define a single average matrix $\btrain=\frac{1}{T_1}\sum_{t=1}^{T_1}B_t$, and solve the following optimization problem. 
\begin{equation}
\label{QP}
F(\bar{B},\lambda) :=
\min_{\phat\ \in
[0,1]
^{n \times m}}\norm{\phat-\btrain}_F^{2} + \lambda\norm{\phat}\tno.
\end{equation}
Here $\lambda > 0$ balances the trade-off between fidelity to $B$ and the low-rank structure.

We prove 
the following generalization bound for $F(\btrain,\lambda)$:
For a matrix $A$ and a distribution $D$ over matrices, let $\ell(A,D)$ be an $L$-Lipschitz measure of the quality of $A$ as a model for $D$. Then, $\forall\hat{\pi} \in [0,1]^{n \times m}$ such that $\norm{\hat{\pi}}\tno\leq \gamma$, with high probability,
\[
|\ell(\hat{\pi},\cD_{\pnull}) - \ell(\hat{\pi},S)| =O\left(\frac{L\gamma}{\sqrt{ T}}\right),
\]
 where $S$ is an i.i.d.~sample of size $T$ drawn from $D_{\pnull}.$
The full proof of this new convergence result can be found in the appendix \ref{ap:analysis}.
We further show that this result holds, in particular, for the Mean Squared Error loss, thus implying that the minimizer of \eqref{QP} over the sample of matrices converges to the best possible stationary model for the given distribution. 

Minimizing \eqref{QP} without requiring $\phat \in [0,1]^{n\times m}$ can be done efficiently, where the result is a model in $\reals^{n\times m}$. This is shown in  \citet{mazumderSpectral}: For a
real-valued matrix $A$, let $\svd(A)$ be the Singular Value Decomposition of
$A$. Let $r$ be the rank of the matrix $\btrain$, and let
$(U,D,V\trn)=\svd(\btrain)$. Then the minimizer of $F(\btrain,\lambda)$ is $UD_{\lambda/2}V\trn$, 
where $D_{\lambda}=[\max(d_{1}-\lambda,0),...,\max(d_{r}-\lambda,0)]$.
Note that for $\lambda\geq 2||\bar{B}||_{sp}$, the minimizer is zero, which provides an upper bound on the valid range for $\lambda$.

We employ this unconstrained minimization approach to get a model   in $\reals^{n\times m}$, and then convert it into a solution $\phat$ that satisfies $\phat\ \in [0,1]^{n \times m}$ using the clipping strategy proposed in \citet{Matrix_Completion_Trace_Norm}.
This approach is much lighter computationally than solving the full constrained minimization, yet it reportedly results in very similar solutions.
This matches our empirical observations in our experiments as well.

The procedure described above for finding a model matrix $\hp$ 
is given in \algoref{hpfind} as the procedure \texttt{FindModel}.

\begin{algorithm}[H]
\caption{\texttt{FindModel}$(\lambda,S)$: Find model matrix}  
\label{alg:hpfind} 
\begin{algorithmic} [1] 
\Require $\lambda > 0$, training data $S = (B_1,\ldots,B_K)$
\Ensure $\hp$
\State  $\btrain \leftarrow \frac{1}{K}\sum_{i=1}^K B_t$.
\State $(U,D,V\trn) \leftarrow \svd(\btrain)$.
\State $\pi' \leftarrow UD_{\lambda/2}V\trn$.
\For{$i \in [m], j \in [n]$}
\State $\hp(i,j) \leftarrow \min(1,\max(\hp'(i,j),0))$. 
\EndFor
\State Return $\hp$.
\end{algorithmic}
\end{algorithm}

Note that while in the rest of our algorithm we use the log-likelihood as a measure of fit between the model $\hp$ and the observed matrix $B_t$,
in \eqref{QP} the Frobenius norm is used instead,
and our generalization bound
holds for the Mean Squared Error. This is because the Frobenius norm is more stable for low-rank approximation, and because optimizing over the log-likelihood under the constraints is significantly more computationally demanding, making it impractical in our setting. Our experiments show that this approach works well in practice.

\subsection{The full training algorithm}\label{sec:full}
In the first step of the training algorithm, which uses the first part of the training set, $S_1$, the value of the regularization parameter $\lambda$ is selected by cross-validation, and  the selected $\lambda$ is used to find the estimated model $\hp$. 
\begin{enumerate}
\item A set of values $\Lambda$ is initialized for cross-validation.
  We use the set $\{ \norm{\btrain}\spn/2^i \}_{i=0}^{K}$, where $K$ is selected adaptively, by identifying when decreasing $\lambda$ further does not improve the log-likelihood on the validation set.
\item $k$-fold cross-validation ($k=10$) is performed to select $\lambda \in \Lambda$: In fold $i$, $S_1$ is divided to a training part $S_1^t(i)$ and a validation part $S_1^v(i)$, and a model $\hp_\lambda(i)$ is calculated by $\hp_\lambda(i) \leftarrow \texttt{FindModel}(\lambda, S_1^t(i))$.
The score of $\lambda$ is set to the average 
\[
L(\lambda) = \frac{1}{k} \sum_{i=1}^k \frac{1}{|S_1^v(i)|}\sum_{B_t \in S_1^v(i)} \logl(B_t, \hp_\lambda(i)).
\]
\item The regularization parameter is set to $\lambda^* \leftarrow \argmax L(\lambda)$. 
\item The estimated model is set to $\hp \leftarrow \texttt{FindModel}(\lambda^*, S_1)$.
  
\end{enumerate}

The second step of the training step uses $S_2$ as follows: Having found a model estimate $\hp$, we now find a regressor $\hat{w}$
for the expected log-likelihood for time-interval $t$.
\begin{enumerate}
	
	\item A training set $\{(v_t,y_t)\}_{t=T_1+1}^T$ for regression is calculated  from $S_2$ as follows:
	\begin{enumerate}
		\item The vector of time-dependent features $v_t$ is calculated using the definition of the features for $t$ (e.g., time-of-day, day-of-weak, etc.)
		\item $y_t \leftarrow \logl(B_t, \hp)$.
	\end{enumerate}
 
\item The regressor is set to 
\[
\hat{w} \leftarrow \argmin_{w \in \reals^d} \sum_{t=T_1+1}^T (y_t - \dotprod{w,v_t})^2.
\] 
\end{enumerate}
The outputs  of the training phase are $\hp$ and $\hat{w}$, where $\hp$ is given as a low-rank matrix decomposition $\hp = U\Lambda V\trn$ of some rank $k$, with $U\in\reals^{n \times k},\Lambda \in\reals^{k \times k},V\in\reals^{m\times k}$, where $\Lambda$ is diagonal. These values are then  used at the anomaly-detection phase to calculate the anomaly score given in Eq.~\eqref{dev}.	
										
\subsection{Unseen objects: cold start}\label{sec:coldstart}
We now explain how we handle the {\em cold start} problem \citep{A_Survey_of_Collaborative_Filtering_Techniques}, which refers to the fact that objects might be observed for the first time \emph{after} the training phase and the estimation of $\hp$. 
The
challenge is to assign
likelihood scores to matrices $B_t$ which include new rows or columns that
do no appear in
$\hp$. 
Previous solutions to the cold start problem in the context of collaborative filtering
have suggested finding an existing user whose pattern of accesses most resembles that of the new user,
and assigning the new user the same prediction as the existing user or a weighted score
of the most similar users \citep{shardanand1995social}. 
Our setup is slightly different, since we are not attempting to predict the values of specific matrix entries. 
In our algorithm, to calculate the log-likelihood of a matrix $B_t$ which includes rows or columns not
present in $\hp$,
we calculate a new version of $\hp$ which extends to these rows and columns. This is based on finding similar users/objects in \emph{latent space}, the low-rank space spanned by $\hp$. 
Each new user or object is projected onto latent space, via a process commonly known as {\em folding} \citep{MR1722790}.
Then, we find its nearest neighbor in the existing $\hp$, based on the distances in latent space. Finally, we assign the new row/column the same probability vector as its nearest neighbor. 
Using distances in latent space reduces the risk of overfitting, and also allows storing and searching over smaller matrices. 

Formally,
let $\btrain_1=\frac{1}{T_1}\sum_{i=1}^{T_1}B_i \in \mathbb{R}^{n \times m}$
be the matrix representing the aggregate access data from the training set $S_1$. Let $\hp=U\Lambda V\trn$ be the rank-$k$ model estimated in the first step of the training phase. Let $G = \btrain_1 V\in\reals^{n \times k}$ and $H = \btrain_1\trn U\in\reals^{m \times k}$.  
Let $G_i$ be the $i$'th row in $G$, and let $H_j$ be the $j$'th row in $H$. These are the latent representations of user $i$ and object $j$ from $\btrain_1$, respectively. \algoref{folded} gives the procedure $\foldlogl$ for calculating the folded log-likelihood of a new observation matrix $B_t$, assuming for simplicity that all new users/objects appear in the last rows/columns of $B_t$ and its dimensions are $n',m'$. The training and anomaly-detection algorithm described above are made to handle the cold start by replacing $\logl$ in the regression learning step and in the anomaly score step with $\foldlogl$.

	\begin{algorithm}[h]	
		\caption{$\mathrm{Folded}\logl(B_t,\hp,G,H,U,V)$}
                \label{alg:folded}
		\begin{algorithmic}  [1] 
			\State $\hp\fold\leftarrow\hp$
			\For{Each row $u_l$ in $B_t$, for $l \in \{n+1,\ldots,n'\}$}
			\State $u' \leftarrow u_l(1:m) \cdot V$ 
			\State$i \leftarrow\argmin_{i}\norm{u'-G_i}_2.$  
			\State Append row $i$ of $\hp\fold$ to the end of  $\hp\fold$. 
			\EndFor
			\For{each column $v_l$ in $B_t$ for $l \in \{m+1,\ldots, m'\}$}
			\State $v' \leftarrow v_l\trn(1:n) \cdot U$  
		\State$j  \leftarrow \argmin_{j}\norm{v' - H_j}_2.$
			\State  Append column  $j$ of $\hp\fold$ to the end of $\hp\fold$.
			\EndFor
			\State \textbf{Return}  $\logl(B_t,\hp\fold)$
		\end{algorithmic}
	\end{algorithm}  

\section{Experiments}\label{sec:experiments}
We tested our algorithm on several data sets. The properties of each data set are given in \tabref{datasets}.

\begin{table}[h]
\begin{center}
\resizebox{\columnwidth}{!}{%
\begin{tabular}{cccccc}
Data Set & Interval & \# intervals & users &objects  \\
&length&in data set \\
\toprule
TDA & 1 hour & $1488$ & $4702$ & $11654$ \\
Amazon & 1 day & $1894$ & $17612$ & $6451$\\
Movielens & 1 day & $1822$ & $29120$ & $24401$\\
Netflix & 1 day & $1565$ & $165405$ & $12938$ \\
TDA (small) & 1 hour & $1488$ & $1000$ & $1000$ \\
Amazon (small) & 1 day & $1894$ & $1000$ & $1000$\\
\bottomrule\\
\end{tabular}
}
\end{center}
\caption{Properties of the tested data sets}
\label{tab:datasets}
\end{table}

The first data set is TDA, which is described in \secref{intro} and is publicly available. TDA records accesses of users to database tables in a live real-world system, during one-hour intervals over a two-month period. 
The second data set is from  Amazon \citep{Lichman:2013}.
It specifies user permissions to resources inside the company during the time period $3.25.05$ --- $8.31.10$. 
The data set specifies which user had permissions to which resource at each day.
We further tested on the movie-rating data sets MovieLens \citep{harper2016movielens} and Netflix \citep{bennett2007netflix}.
In these two data sets, the objects are movies, and an access occurs when a user rates a movie.
It should be noted that while no single user-movie pair is repeated in the movie-rating data sets, anomalous behavior (e.g.~users rating movies of a genre they seldom rate) can still be identified using latent-factor analysis as performed by our algorithm.
We used MovieLens data from the years 2010-2014, during which the level of activity was fairly stable. We used Netflix data from the dates $12.8.99$ --- $4.19.14$, for which complete data was available.

The available data sets do not contain known anomalous accesses. Thus, in our experiments we
injected anomalous behavior into random intervals, as explained below.
We compare our algorithm (termed TLR in the figures below) to five baselines: First, the four algorithms described in \secref{rel}: SRMF \citep{Roughan:2012:SCS:2369156.2369159}, PCA \citep{Lakhina:2004:DNT:1015467.1015492} ,
RobustPCA \citep{Candes:2011:RPC:1970392.1970395} and MET \citep{TTB_MET}.
A fifth baseline, which we call MEAN, is similar to our algorithm, except that the deviation score is calculated with respect to the mean log-likelihood of the regression training set, without any adjustments based on regression. 
In each of the first four baselines, we used the default parameters as recommended by the authors.
For SRMF, PCA and MET, the anomaly scores were generated by evaluating the norm of the residual score on each interval. 
For RobustPCA, the anomaly score was generated by evaluating the norm of the sparse component of the interval.

The first four baseline algorithms all process the entire data set at once.
As a result, due to the size of our data sets, it was impossible to run these baselines on the full data sets on a reasonably high-capacity multicore server. Indeed, an important advantage of our algorithm is that it does not process the entire data set at once, and thus can handle much higher-dimensional data sets without requiring a large memory. Since we could not run these baselines on the full data sets, we ran a full comparison on a down-sampled version of the TDA and Amazon data sets, each including only 1000 users and 1000 objects, selected at random from each data set: this size was the largest that was feasible for all algorithms. We term these data sets below ``TDA (small)'' and ``Amazon (small)''. Down-sampling the Movielens and Netflix data set proved ineffective, since the result was so sparse that all algorithms failed completely. For the full data sets, we report the results of our algorithm and of the MEAN baseline.

For our algorithm, we used the following natural time-dependent features for regression,
inspired by \citet{5694074}:
A binary ``weekend'' feature,
the log-likelihood of the previous interval and of the one $24$ hours ago (for TDA) or  a week ago (for the others), the number of accesses in the current interval, the number of intervals since the last training set interval, day-of-the-week, and for TDA also hour of the day $h \in \{1,\ldots,24\}$ and shifted hour of the day ($(h + 12) \bmod 24$). 

\textbf{Accuracy of regression.}
\figref{regression} shows, for each full data set, the true log-likelihood of each test interval against the predicted log-likelihood based on the learned regressor $\hat{w}$. 
A straight diagonal line would indicate a perfect prediction.
Indeed, the prediction is quite successful for these data sets, and the correlation coefficients ($\rho$) are very close to one,
indicating that using linear regression here is reasonable. 
\begin{figure}[h] 
{\centering
	\hide{
 \includegraphics[width = 0.23\textwidth]{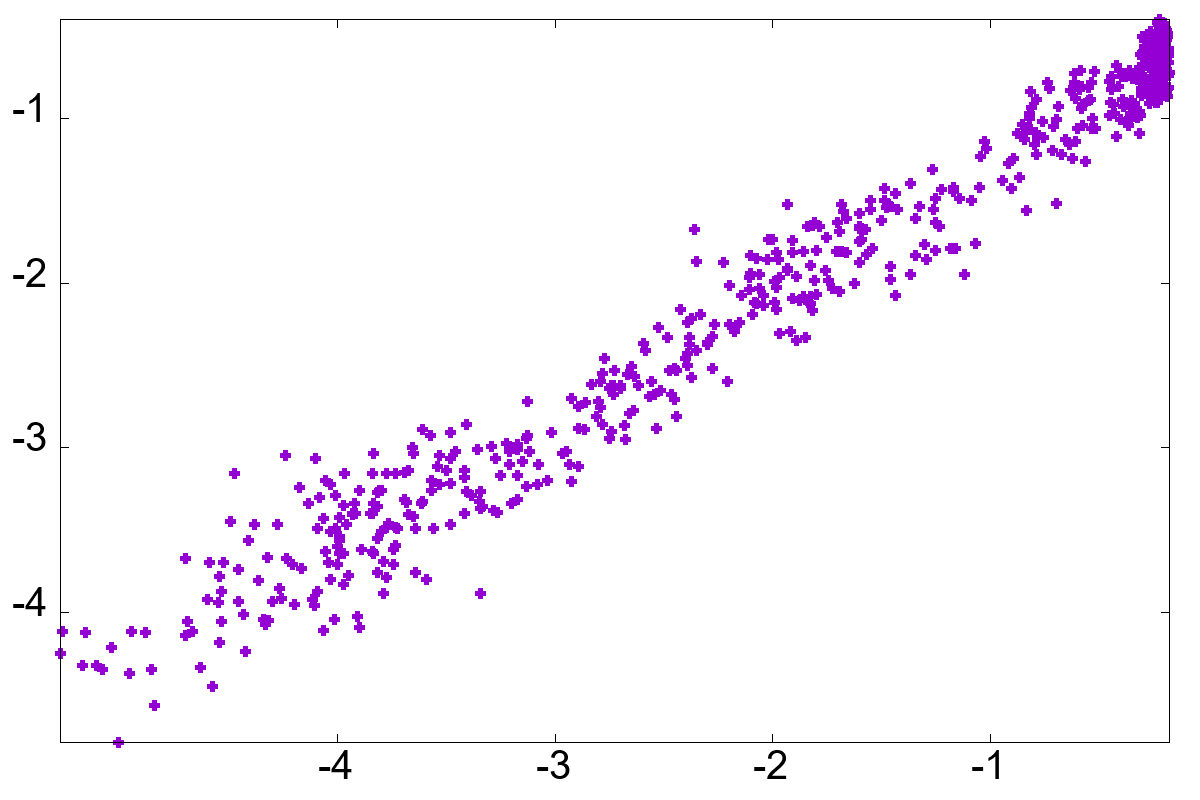}
 	 \includegraphics[width = 0.23\textwidth]{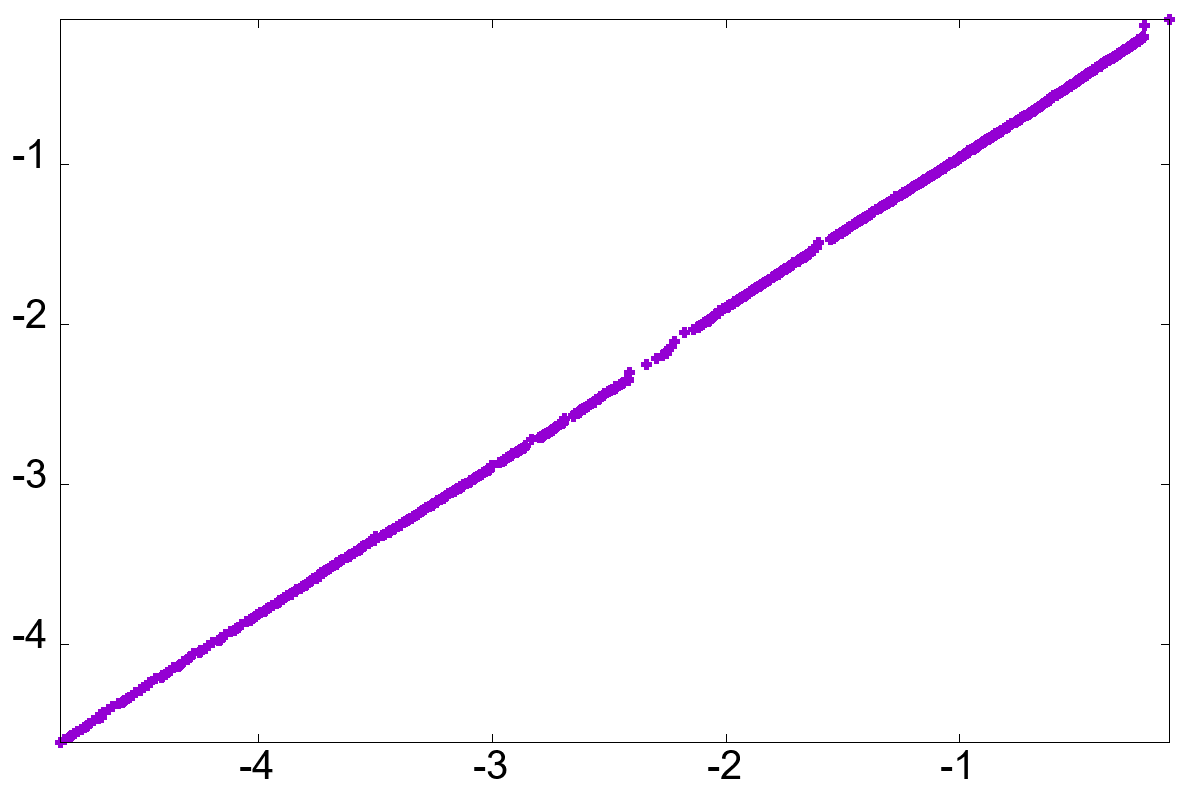}\\
	\includegraphics[width = 0.23\textwidth]{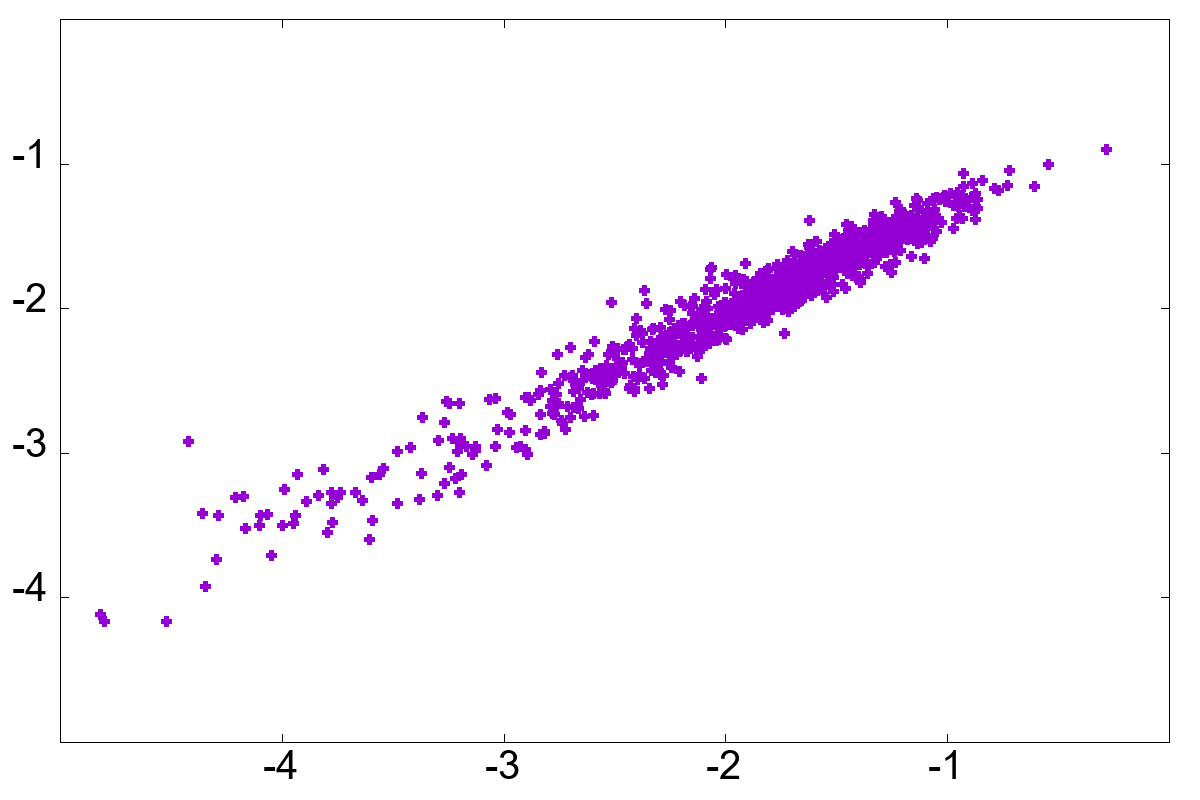}
		\includegraphics[width = 0.23\textwidth]{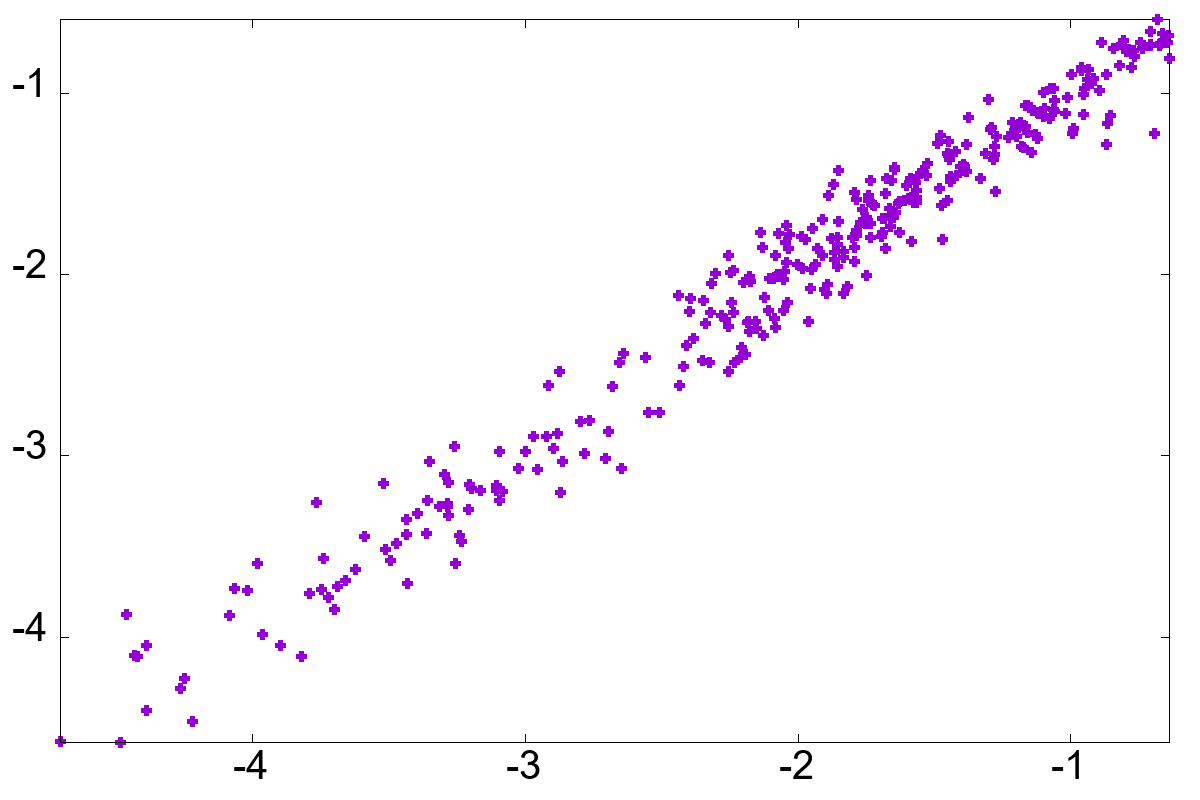}\\
		
		\includegraphics[width=0.23\textwidth]{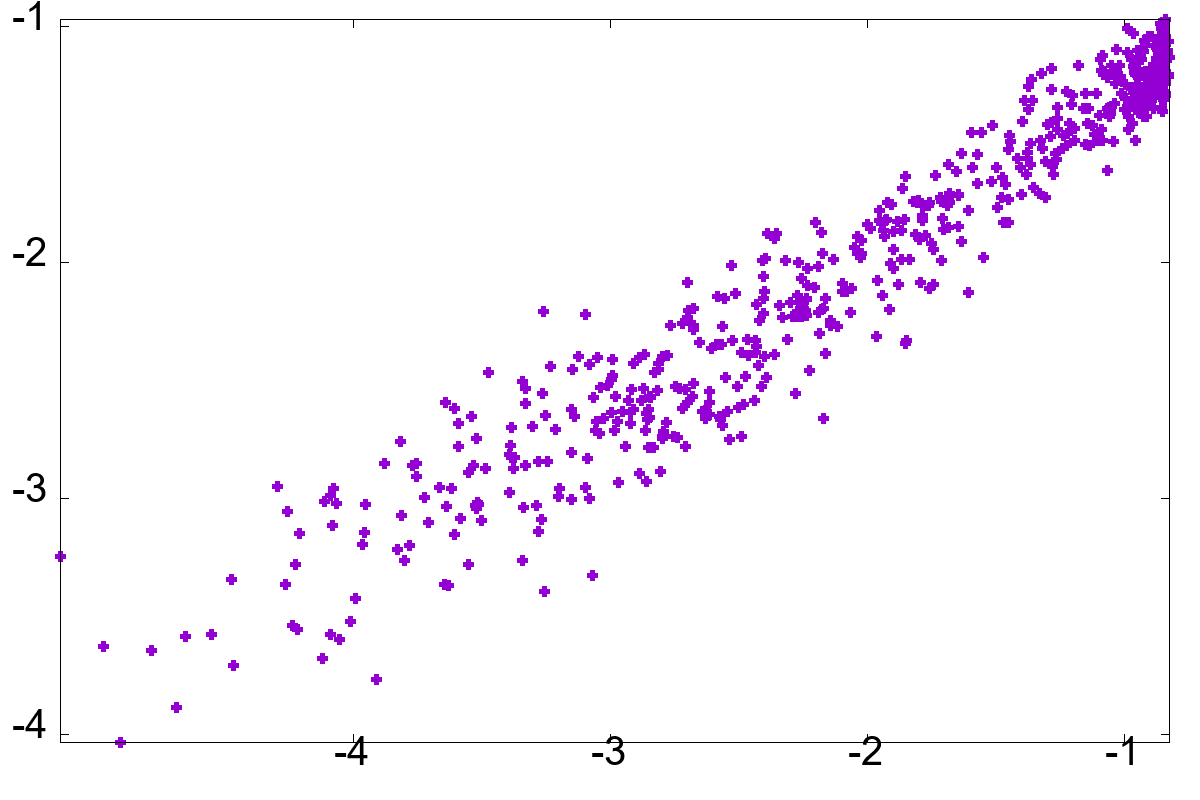}
		\includegraphics[width=0.23\textwidth]{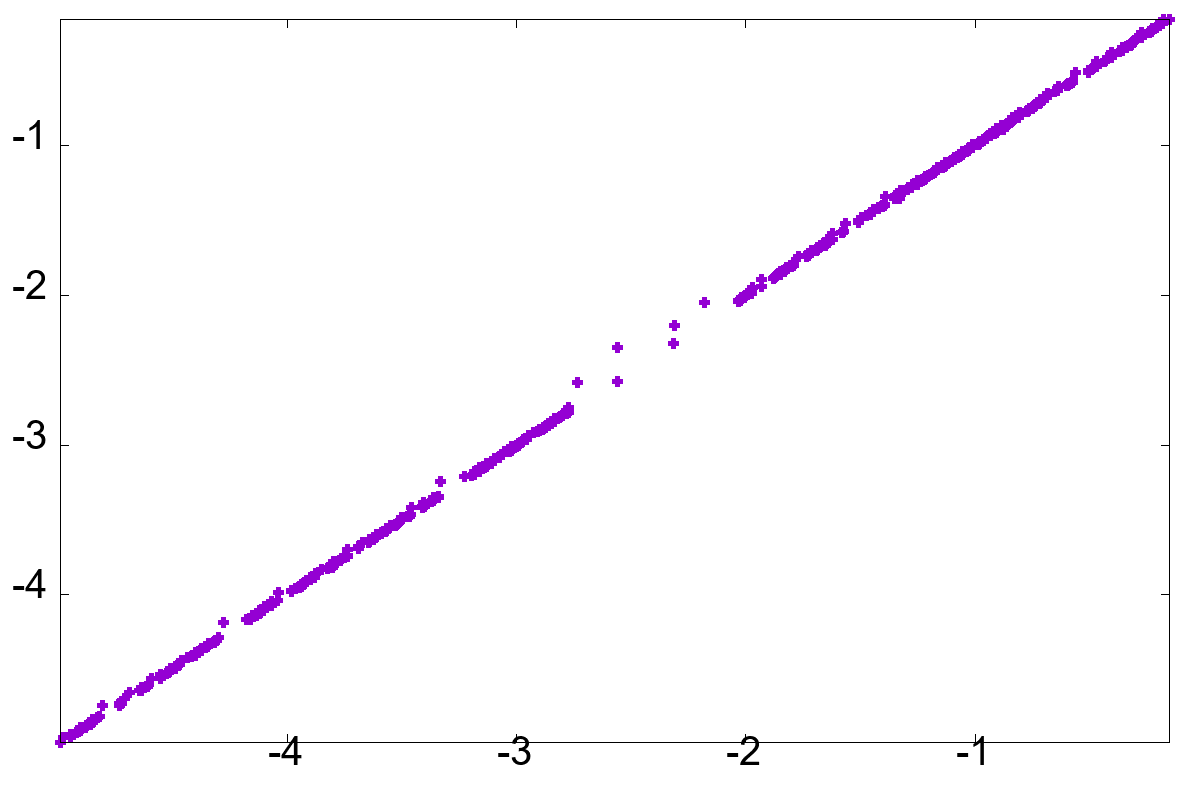}\\
	}
		\includegraphics[width=0.23\textwidth]{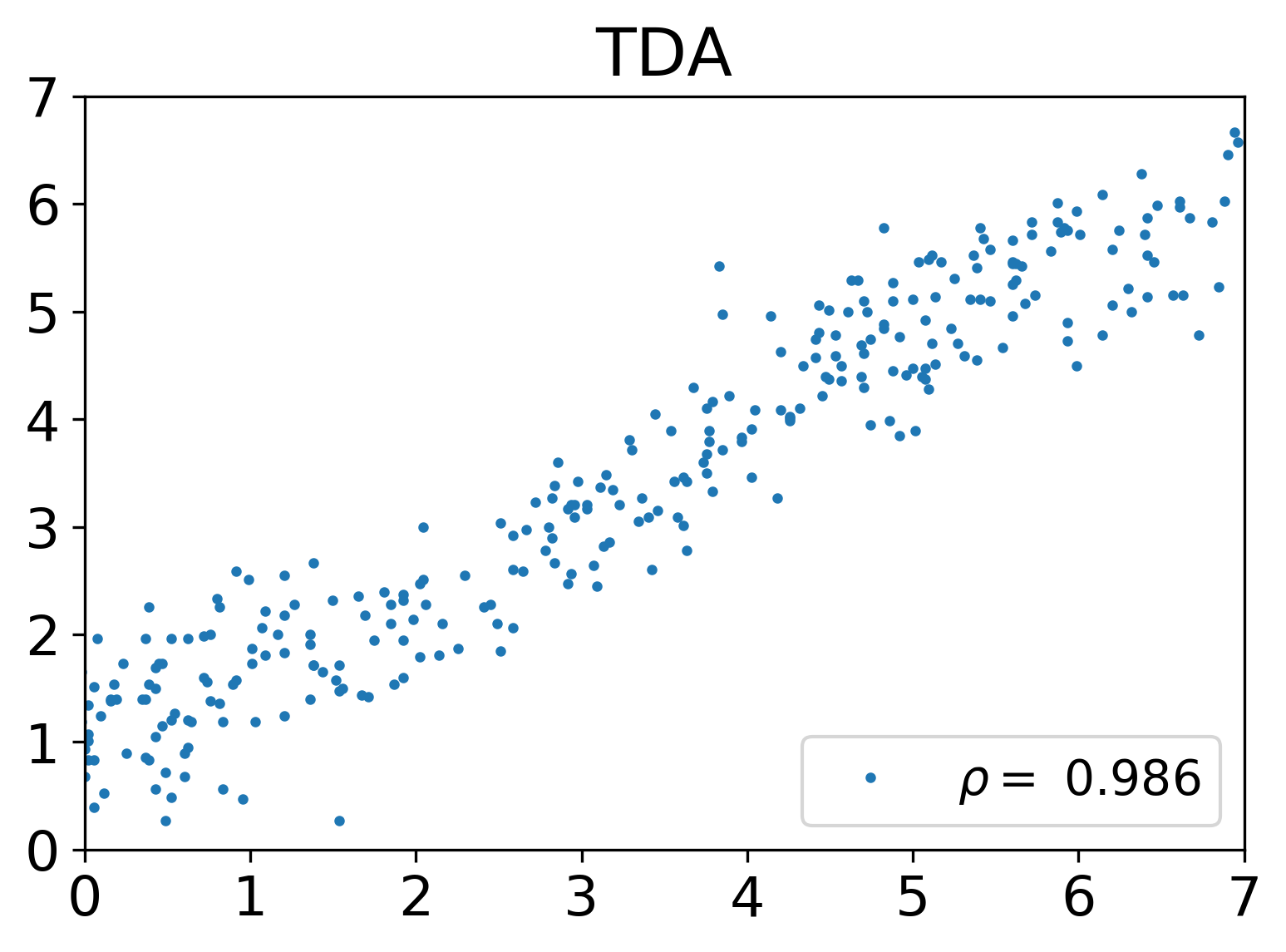}
			\includegraphics[width=0.23\textwidth]{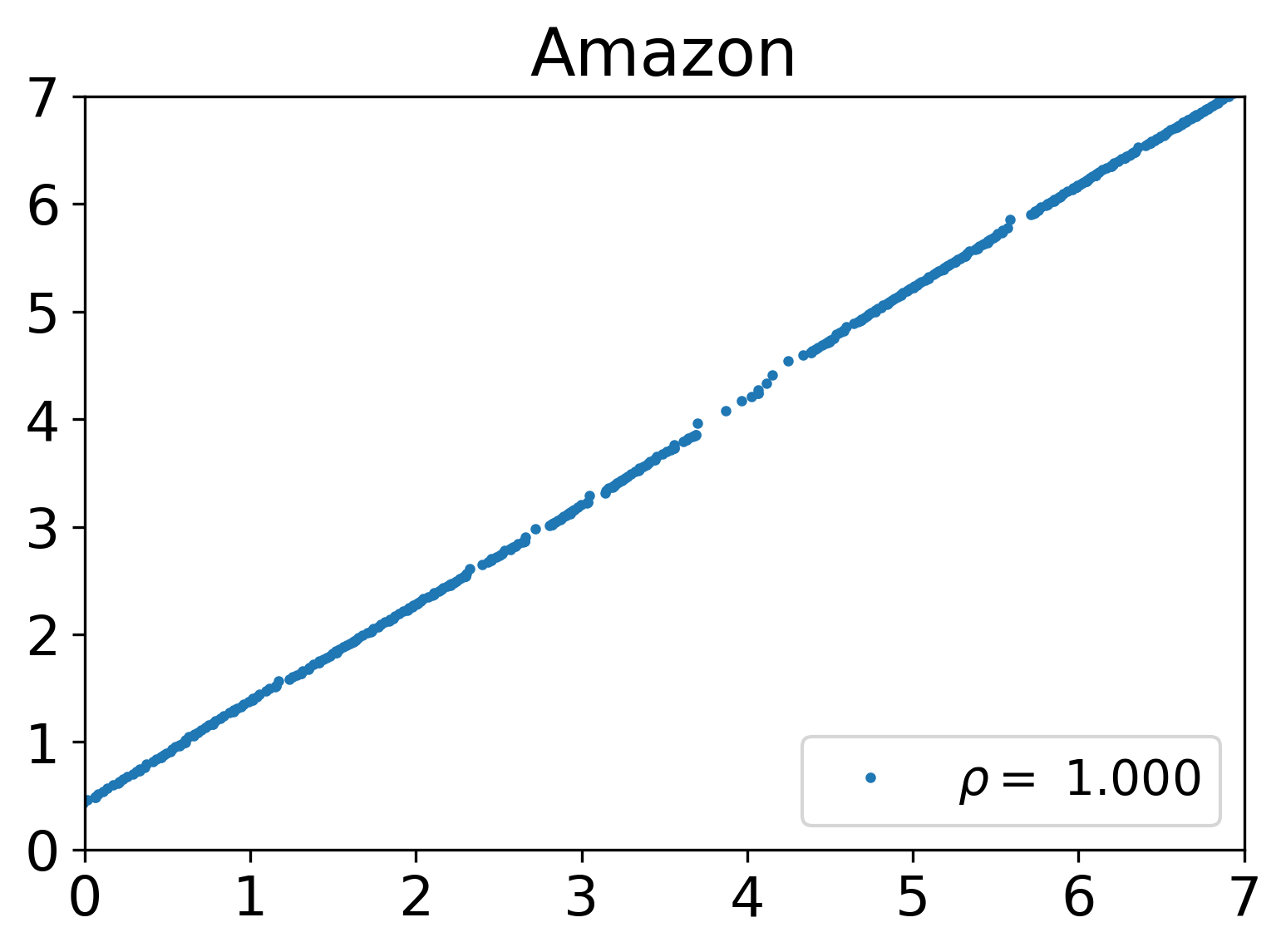}\\
				\includegraphics[width=0.23\textwidth]{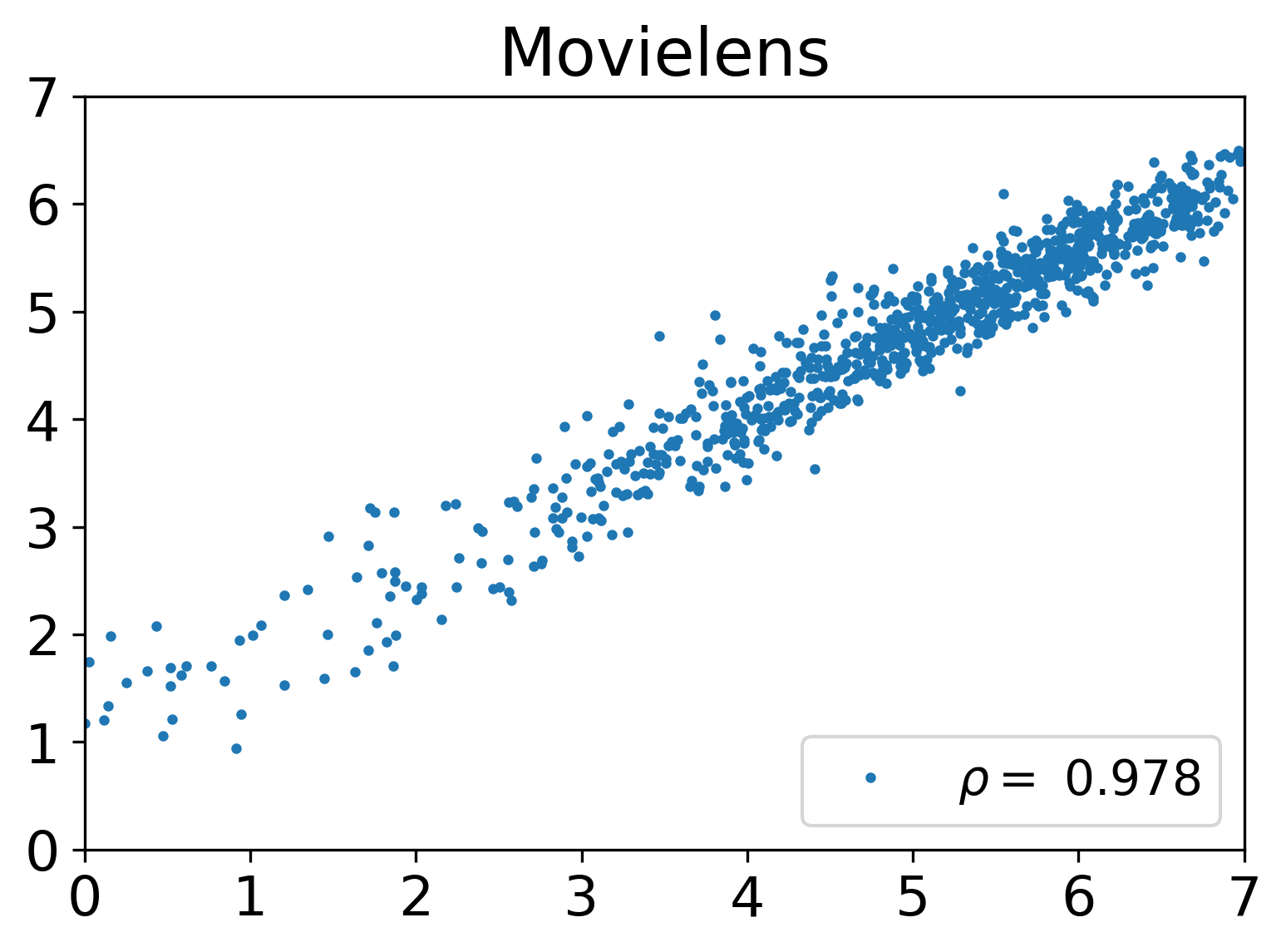}
					\includegraphics[width=0.23\textwidth]{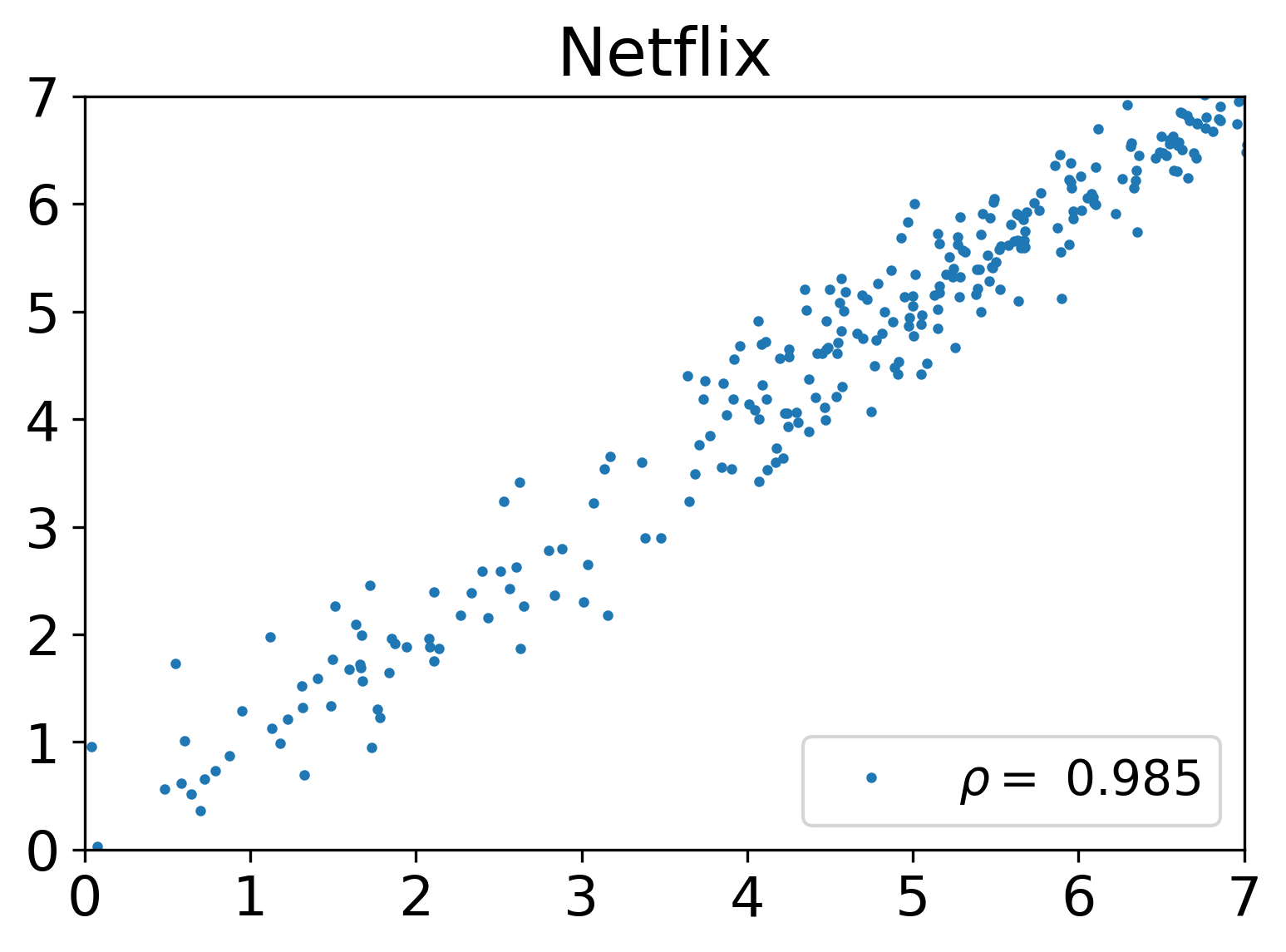}\\
					\includegraphics[width=0.23\textwidth]{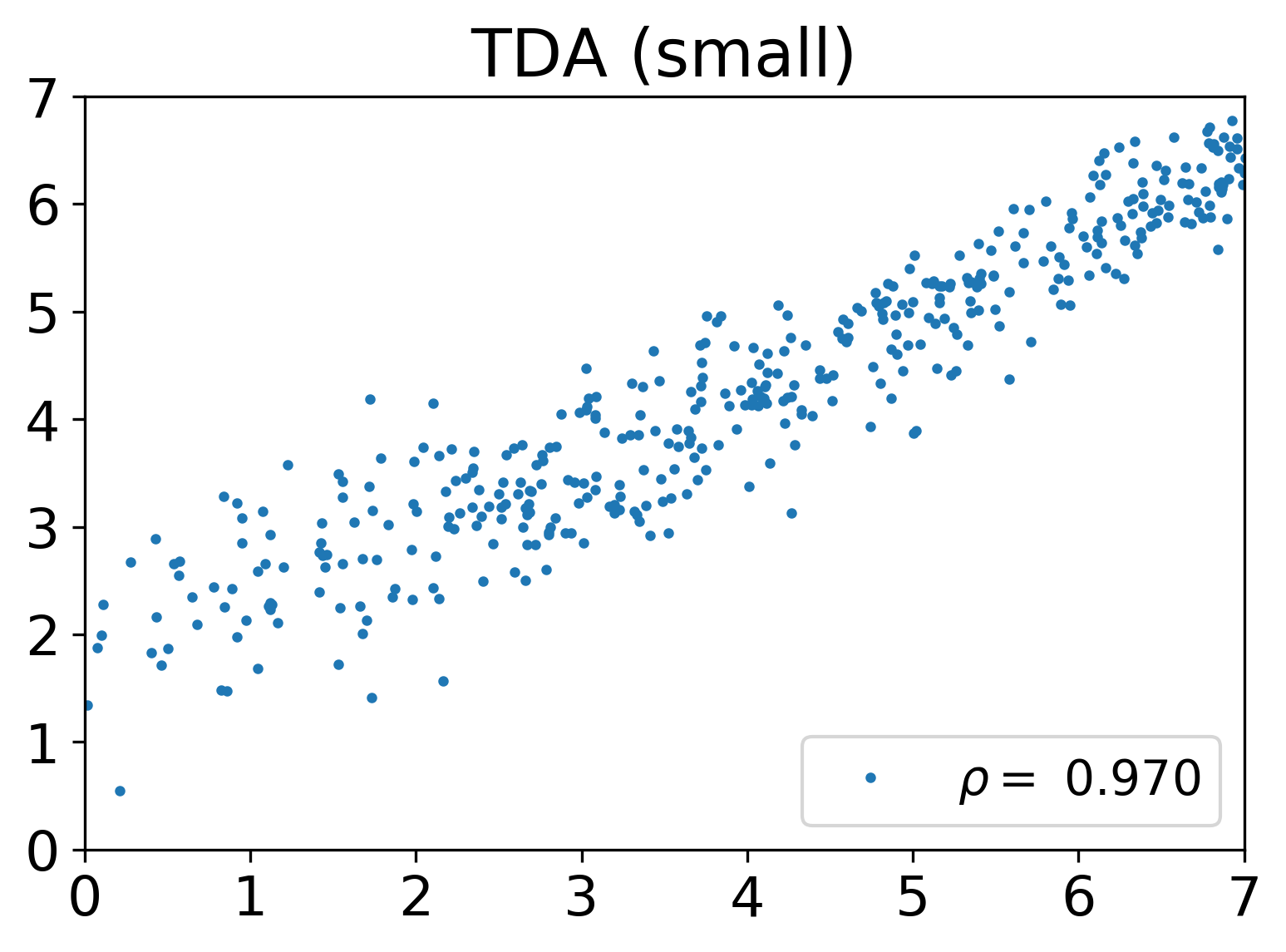}
					\includegraphics[width=0.23\textwidth]{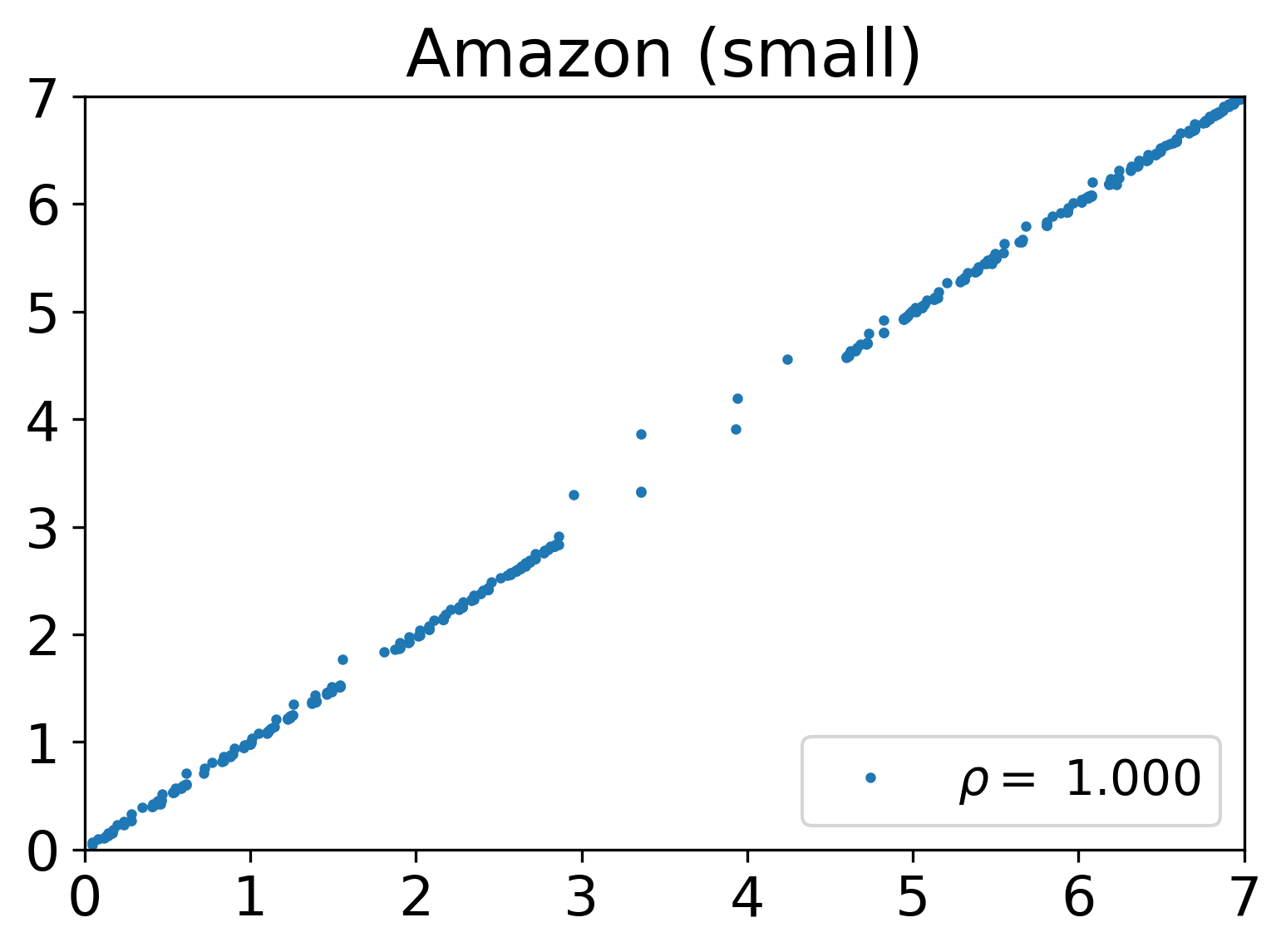}
		
 	\caption{True vs.~predicted log-likelihood on the test set (in arbitrary units). 
 	 }
\label{fig:regression}
}
 \end{figure}

 \textbf{Experiment I: Random accesses.} In this experiment, an anomalous interval is simulated by adding random accesses to it: Each bit
 in the interval's access matrix is changed to $1$ with an independent
 probability of $\epsilon >0$. We ran the algorithms 100 times on each data set, each time randomly selecting a single interval to simulate as anomalous. 
 For each noise level, we calculated the AUC (Area under the curve) of the combined ROC curve, and plotted it against the value of $\log_{10}(\epsilon)$. 
The results of our algorithm and of MEAN on the full data sets, as well as the results of all algorithms on the down-sampled data sets, are reported in \figref{noise}. It can be seen that the regression model improved the identification of the anomaly in a wide range of noise levels, and that our algorithm is usually better than all baselines.

 \figurenoise
\figurehistswitches

 \textbf{Experiment II: Accesses at an anomalous time.} We simulated a
 behavior which is normal at one time, but possibly anomalous at a different
 time, by replacing two randomly chosen intervals with each other. 
 This moves intervals to a time in which they might be unexpected, hence should be identified as anomalies. We ran the algorithms 100 times on each data set, each time replacing intervals of a single random pair. \figref{rocswitches} shows the ROC curves generated from the anomaly scores of each of the algorithms. Note that if the times of the intervals in the pair are similar, e.g., both in the morning of a workday, then no anomaly should be identified. Thus, even the best algorithm could have an AUC which is not very close to one. \figref{rocswitches} shows that  our algorithm is the most successful here. With the exception of SRMF, the other algorithms are not better than random guessing on this task. This should not be surprising, since these algorithms do not take into account the dependence on the timing of events.

\paragraph{Run time and memory} The baseline algorithms operate on a single matrix/tensor which includes the entire data set, and with no distinction between training time and test time.
 As a result, the memory requirements of these algorithms are prohibitive for large data sets. In contrast, our algorithm processes single-interval matrices one by one. Therefore, it  requires significantly less memory. In addition, our algorithm performs matrix optimization only during training time, 
while the calculation of anomaly scores during test time is fast, and requires a small memory footprint for the trained model. This allows real-time anomaly detection as new matrices appear. \tabref{timeAnalysis} reports the run time of each of the algorithms on the tested data sets. The baseline algorithms could only be run on the reduced data sets due to their memory requirements. Our approach shows a clear run-time advantage, during training and more so during testing. 

\tableruntime

\section{Conclusions}
The experiments demonstrate that our approach obtains superior results to previous algorithms, while requiring significantly less computational resources. While we focused here on identifying anomalous time intervals, this approach can be adapted to identifying specific users or objects which are anomalous.
An additional important challenge is to develop a streaming version of the training stage.
These adaptations will be studied in future work.

\section{Acknowledgements}
This work was supported in part by an IBM grant, and by the Israel Science Foundation (grants~555/15, 755/15).
 
\fontsize{9pt}{10pt} \selectfont

\bibliographystyle{aaai}

\appendix

\section{Appendix: Convergence analysis}
 
\label{ap:analysis}

For two real-valued matrices $X,Y \in \reals^{n \times m}$, define the mean squared error between $X$ and $Y$ by 
\[
\mse(X,Y) :=\frac{1}{mn} \norm{X - Y}_F^2.
\] 

Recall that $\pnull$ is the assumed stationary model, while $\hp$ is our estimate of that stationary model.
Below we show that by finding $F(\btrain, \lambda)$, as defined in \eqref{QP}, we are effectively minimizing
$\mse(\hp,\pnull)$ up to a
term that decays to zero.
\eqref{QP} minimizes the squared Frobenius norm $\norm{\hp - B}_F^2$ with an additional regularization term $\lambda \norm{\hp}\tno$. This is equivalent to minimizing $\norm{\hp - B}_F^2$ subject to a constraint on $\norm{\hp}\tno$. 
Moreover, if $S = (B_1,\ldots,B_T)$ and $\btrain = \frac{1}{T}\sum_{t=1}^{T}B_t$, we have
\begin{align*}
&\norm{\hp - \btrain}_F^2 = \norm{\hp}_F^2 - 2\dotprod{\hp,\btrain} + \norm{\btrain}_F^2 =\\ &\frac{1}{T}\sum_{t=1}^{T} \norm{\hp - B_t}_F^2 - \frac{1}{T}\sum_{t=1}^{T} \norm{B_t}_F^2 + \norm{\btrain}_F^2.
\end{align*}
Therefore, 
\begin{align}
&\norm{\hp - \btrain}_F^2 = \frac{1}{T}\sum_{t=1}^T\norm{\hp - B_t}_F^2 +C\\
&= mn\cdot \frac{1}{T}\sum_{t=1}^T\mse(\hp,B_t) + C,
\end{align}

where $C$ is a constant independent of $\hp$. 
By a similar derivation, for another constant $C'$ independent of $\hp$,
\[
\E_{B \sim \pnull}[\mse(\hp,B)] = \mse(\hp,\pnull) + C'.
\] 
Therefore, if $|\E_{B \sim \pnull}[\mse(\hp,B)]  - \frac{1}{T}\sum_{t=1}^T\mse(\hp,B_t)|$ is small, finding $F(\btrain,\lambda)$ is a good proxy for minimizing $\mse(\hp,\pnull)$. We prove a more general claim, which bounds the diffrence between the empirical loss and the true loss of a general Lipschitz loss, which we define below. Our result generalizes a result from \cite{PU_Learning_for_Matrix_Completion} to the case of a
  training sample with several matrices. The proof employs techniques from \cite{PU_Learning_for_Matrix_Completion} and \cite{Matrix_Completion_Trace_Norm}. 

Let $\ell:[0,1] \times [0,1] \rightarrow \reals_+$ be a loss function. For two matrices $A,B \in [0,1]^{n \times m}$, let 
\[
\ell(A,B) := \frac{1}{nm} \sum_{i \in [n], j\in [m]}\ell(A(i,j), B(i,j)).
\]
For $L > 0$, $\ell$ is $L$-Lipschitz in the first argument if 
\[
\forall x,x',y\in \reals, |\ell(x,y) - \ell(x',y)| \leq L|x -x'|,
\]
with an analogous definition
for the second argument.

For a matrix $X\in[0,1]^{n \times m}$ and a distribution $P$ over matrices $[0,1]^{n \times m}$;
denote the \emph{true loss} of $X$ by $\ell(X,P) := \E_{Y \in P}[\ell(X,Y)]$. For  a (multi)set $S \subseteq [0,1]^{n \times m}$, denote the \emph{empirical loss} of $X$ on $S$ by $\ell(X,S) := \frac{1}{|S|}\sum_{Y \in S}[\ell(X,Y)]$.

The theorem below gives a guarantee for $L$-Lipschitz losses. 
Note that for the squared-loss $\ell := \mse$, where  $\mse(x,y):= (x-y)^2$, $\ell$ is $2$-Lipschitz in both arguments: 
For $x,x',y \in [0,1]$, 
\begin{align*}
& |(x-y)^2-(x'-y)^2|=|x^{2}-x'^{2} -2xy+2x'y|\\
&= |(x-x')(x+x')-2y(x-x')| =|(x-x')(x+x'-2y)|\\
&= |(x-x')||(x+x'-2y)|\leq 2|x-x'|.
\end{align*}
Thus the theorem below holds for the $\mse$ loss with $L = 2$, proving that as the size of the training set grows, minimizing \eqref{QP} converges to a minimization of $\mse(\hp,\pnull)$ with high probability. In the following theorem and proof, we use $c$ to indicate a universal constant, whose value can change from line to line.
\begin{theorem}
Let $\ell(x,y)$ be an $L$-Lipschitz loss.
Let $\pnull \in [0,1]^{n \times m}$, and let $\cD_\pnull$ be the distribution satisfying the probabilistic model defined in \secref{approach}. Assume w.l.o.g.~that $m \geq n$. 
Let $S = (B_1,B_2,\ldots,B_T) \sim \cD_\pnull^m$ be an i.i.d.~sample from $\cD_\pnull$.
With a probability of at least $1-\delta$, for all $\pi \in [0,1]^{n \times m}$ such that $\norm{\pi}\tno\leq \gamma$, 

  \[
  |\ell(\pi,\cD_{\pnull}) - \ell(\pi,S)| \leq c\frac{L\gamma}{n\sqrt{mT}}+L\sqrt{\frac{\ln(\frac{2}{\delta})}{2Tmn}}.
  \]
\label{thm:convergence}
\end{theorem}

\begin{proof}
 
Let $\Pi \subseteq [0,1]^{n \times m}$. Denote $\psi(S) := \sup_{\pi \in \Pi} (\ell(\pi,\cD_{\pnull}) - \ell(\pi,S))$.
By definition,
\begin{equation}\label{psi}
 \forall \pi \in \Pi,\quad\ell(\pi,\cD_{\pnull})\leq \ell(\pi,S)+\psi(S).
 \end{equation}
Note that the sample $S$ has $Tnm$ independent (though not identically distributed) entries. Thus, we upper-bound $\psi(S)$ with high probability using McDiarmid's inequality \cite{mcdiarmid1989method}, which states that if for every two samples $S,S'$ which differ by a single entry, $|\psi(S) - \psi(S')| \leq \alpha$, then
\[
\P \left[\psi(S)-\E[\psi(S)] \geq \epsilon\right]\leq e^{\frac{-2\epsilon^{2}}{Tnm\alpha^{2}}}.
\]
To bound $\alpha$, consider two samples, $S,S'$, where the matrices in $S$ are denoted $B_t$ and the matrices in $S'$ are denoted $B'_t$. Suppose that $S'$ differs from $S$ by a single entry $(i,j)$ in the matrix $B_{t_o}$, such that $B_{t_o}(i,j) = 1$ and $B'_{t_o}(i,j) = 0$. We have 
\begin{align*}
&|\psi(S)-\psi(S')|= \\
&|\sup_{\pi}(\ell(\pi,\cD_{\pnull}) - \ell(\pi,S)) - \sup_{\pi}(\ell(\pi,\cD_{\pnull}) - \ell(\pi,S'))|
\end{align*}

For bounded    
 $f,g:\reals\rightarrow \reals$, it holds that 
 $|\sup_{x}f(x)-\sup_{x}g(x)|\le|\sup_{x}(f(x)-g(x))|$.
Therefore
\begin{align*}
&|\psi(S)-\psi(S')| \leq \frac{1}{T}|\sup_{\pi} \sum_{t}(\ell(\pi,B_t) - \ell(\pi,B'_t))| \\
&= \frac{1}{Tmn}|\sup_{\pi}
(\ell(\pi(i,j),B_{t_o}(i,j))  - \ell(\pi(i,j),B'_{t_o}(i,j))|\\
&\leq \frac{L}{Tmn}|B_{t_o}(i,j)-B'_{t_o}(i,j)]| \leq \frac{L}{Tmn},
\end{align*}
Where we used the fact that $\ell$ is $L$-Lipschitz in the second argument.
Setting $\alpha = L/(Tnm)$ and applying McDiarmid's inequality, we get 
 \[
 \P[\psi(S)-\E[\psi(S)]\geq \epsilon]\leq  e^{-2\epsilon^2 Tmn/L^2}.
 \]
 Setting the latter to
 $\frac{\delta}{2}$, we obtain that
 with a probability of at least $1-\frac{\delta}{2}$,
\begin{equation}
\label{hi-prob}
 \psi(S) \leq \E[\psi(S)]+L\sqrt{\frac{\ln(\frac{2}{\delta})}{2Tmn}}.
\end{equation} 

It remains
to bound $\E[\psi(S)]$. 
 Let $S' = (B'_1,\ldots,B'_T) \sim \cD_\pnull^T$ be a sample which is independent from $S$.
Following a standard symmetrization argument as in \cite{PU_Learning_for_Matrix_Completion}, we have 
\begin{align*}
& \E[\psi(S)] =\E_S[\sup_{\pi}(\ell(\pi,\cD_{\pnull})-\ell(\pi,S))]]\\
&= \E_{S}[\sup_{\pi}(\E_{S'}[\ell(\pi,S')-\ell(\pi,S)])]\\
&\leq \E_{S,S'}[\sup_{\pi}(\ell(\pi,S')-\ell(\pi,S))]\\
&= \frac{1}{mnT}\E_{S,S'}\big[ \sup_{\pi}\sum_{i,j,t}\ell(\pi(i,j),B'_t(i,j))\\
&\qquad\qquad-\ell(\pi(i,j),B_t(i,j))\big].
\end{align*}
Letting $\sigma_{i,j,t} \sim \mathrm{Uniform}\{-1,1\}$ for $i\in [n],j\in[m],t\in[T]$ be independent Rademacher variables, it follows that
\begin{align}
\nonumber
 &\E[\psi(S)] =\\
 &\frac{1}{mnT}\E_{S,S',\sigma}\Big[ \sup_{\pi} \sum_{i,j,t}\sigma_{i,j,t}\big(\ell(\pi(i,j),B'_t(i,j))\notag\\
&\qquad\qquad\qquad-\ell(\pi(i,j),B_t(i,j)) \big)\Big]\\
&\leq \frac{2}{mnT}\E_{S,\sigma}\left[ \sup_{\pi}\sum_{i,j,t}\sigma_{i,j,t}\ell(\pi(i,j),B_t(i,j))\nonumber\right]\\
& \leq \frac{2L}{mnT}\E_{\sigma}\left[ \sup_{\pi(i,j)}\sum_{i,j,t}\sigma_{i,j,t}\pi(i,j). \label{psi-rade}\right]
\end{align}
The last inequality follows from Talagrand's contraction principle \citep{LedouxTal91},
together with the fact that $\ell(x,y)$ is $L$-Lipschitz in
both arguments. 

Now, denoting $\nu_{i,j} := \sum_t \sigma_{i,j,t}$,
we observe that
\begin{align*}
&\E_{\sigma}\left[ \sup_{\pi}\sum_{i,j,t}\sigma_{i,j,t}\pi(i,j)\right]=\\
&\E_{\sigma}\left[ \sup_{\pi}\sum_{i,j}\pi(i,j)\sum_{t}\sigma_{i,j,t}\right]=\E_{\sigma}\left[ \sup_{\pi}\sum_{i,j}\pi(i,j)\nu_{i,j}\right]
\end{align*}
Since the nuclear and spectral norms are dual \citep{MR2978290}, a matrix H\"older inequality holds, where $\nu$ is the matrix with entries $\nu_{i,j}$: 
\[
\sum_{i,j}\pi(i,j)\nu_{i,j}\leq\norm{\pi}\tno\norm{\nu}_{\spn}.
\]
Therefore, combining with \eqref{psi-rade}, and assuming $\sup_{\pi \in \Pi} \norm{\pi}\tno \leq \gamma$, it follows that
 \begin{equation}\label{boundnu}
\E[\psi(S)] \leq \frac{2L}{mnT}\E_{\sigma}[\sup_{\pi \in \Pi} \norm{\pi}\tno\norm{\nu}_{\spn}] \leq  \frac{2L\gamma}{mnT}\E_{\sigma}[\norm{\nu}_{\spn}].
\end{equation}

It remains to bound $\E_\sigma\left[ ||\nu||_{\spn}\right]$.
To this end, recall the following result of \citet{latala2005some}:
There is a universal constant $c>0$ such that for any random matrix 
$Z\in \reals^{n\times m}$ with independent mean-zero entries, we have
\begin{align*}
 &\E\left[ ||Z||_{\spn}\right]\leq \\
 &c \Big(\max_{i}\sqrt{\sum_{j}\E\left[  Z(i,j)^{2}\right]}\\
&\qquad +\max_{j}\sqrt{\sum_{i}\E \left[ Z(i,j)^{2}\right]}+ \sqrt[4]{\E \left[  Z(i,j)^{4}\right]}\Big).
 \end{align*}
 \normalsize
Now $\nu_{i,j}$ is a sum of $T$ i.i.d.~Rademacher random variables, and thus $\E[\nu_{i,j}]=0$
and $\E[\nu_{i,j}^2]=T$.
To bound the $4$th moment, we appeal to Khinchine's inequality \citep[Exercise 5.10]{MR3185193}, which states that $\E[\nu_{i,j}^{4}]
\leq
\frac{4!}{2\cdot2^{2}}
\E[\nu_{i,j}^{2}]^2
=
3T^{2}
.
$
Substituting these moment bounds into Lata{\l}a's bound, we get
 \[
\E \left[ ||\nu||_{\spn}\right]\leq c\paren{\sqrt{nT}+\sqrt{mT}+ \sqrt[4]{nm}\sqrt{T}}.
 \]
Since $\sqrt[4]{nm}\leq \frac{1}{2}(\sqrt{n}+\sqrt{m})$ and $m \geq n$,
we have $ \E\left[ ||\nu||_{\spn}\right]\leq c\sqrt{Tm}$.  
Substituting this into \eqref{boundnu}, we get
 \[
 \E[\psi(S)] \leq c\frac{L\gamma}{n\sqrt{mT}}.
 \]
 
 Combining this with \eqref{hi-prob} yields that
 with probability of at least $1-\frac{\delta}{2}$,
 \[
 \psi(S) \leq c\frac{L\gamma}{n\sqrt{mT}}+L\sqrt{\frac{\ln(\frac{2}{\delta})}{2Tmn}}.
 \]
 Combining this with \eqref{psi}
 and
 the observation that
\begin{align}
 &\forall \pi \in \Pi ~~
 \ell(\pi,\cD_{\pi})\\
 &\leq \ell(\pi,S)+\sup_{\pi' \in \Pi}(\ell(\pi',\cD_{\pi})-\ell(\pi',S)) \\
 &\leq \ell(\pi,S)+\psi(S)
\end{align}
 ,
we get that if $\sup_{\pi \in \Pi} \norm{\pi}\tno \leq \gamma$, then 
with probability of at least $1-\frac{\delta}{2}$,
 \[
 \forall \pi \in \Pi, \quad \ell(\pi,\cD_{\pnull})-\ell(\pi,S)\leq c\frac{L\gamma}{n\sqrt{mT}}+L\sqrt{\frac{\ln(\frac{2}{\delta})}{2Tmn}}.
 \]
 An analogous argument yields the same bound for $\ell(\pi,S)-\ell(\pi,\cD_\pnull)$, which implies the statement of the theorem via a union bound.
\end{proof}
 
\end{document}